\documentclass{elsarticle} %
\makeatletter
\def\ps@pprintTitle{%
 \let\@oddhead\@empty
 \let\@evenhead\@empty
 \def\@oddfoot{}%
 \let\@evenfoot\@oddfoot}
\makeatother

\usepackage{float}
\usepackage[labelformat=simple]{subcaption}

\usepackage{amsfonts}
\usepackage{color}
\usepackage{comment}
\usepackage{enumerate}
\usepackage[letterpaper, top=2.5cm, bottom=2.5cm, left=2.2cm, right=2.2cm]{geometry}
\usepackage{times}
\usepackage{amsmath}
\usepackage{changepage}
\usepackage{amssymb}
\usepackage[amsmath,thmmarks]{ntheorem}
\usepackage{graphicx}
\newtheorem{theorem}{Theorem}[section]

\newtheorem{corollary}[theorem]{Corollary}

\newtheorem{lemma}[theorem]{Lemma}

\newtheorem{proposition}[theorem]{Proposition}

\newenvironment{proof}[1][Proof]{\textbf{#1.} }{\ \rule{0.5em}{0.5em}}

\newcommand{\R}{\mathbb{R}}

\def\FISHER#1{\mathbf{F}({#1})}

\begin{document}
\begin{frontmatter}
\title{Scalable Information Inequalities for Uncertainty Quantification} 
\author[UMASS]{Markos A. Katsoulakis}\ead{markos@math.umass.edu}
\author[UMASS]{Luc Rey-Bellet}\ead{luc@math.umass.edu}
\author[UMASS]{Jie Wang}\ead{wang@math.umass.edu}
\address[UMASS]{Department of Mathematics and Statistics, University of Massachusetts Amherst, Amherst, MA 01003, USA}  
\date{\today}
\begin{abstract}
In this paper we demonstrate the only available scalable information bounds for quantities of interest of  high dimensional probabilistic  models. 
Scalability of inequalities   allows us to (a)  obtain uncertainty quantification  bounds for quantities of interest  in the large degree of freedom  limit and/or at long time regimes; (b)  assess the impact of large model perturbations as in nonlinear response regimes  in statistical mechanics; (c) address model-form uncertainty, i.e. compare different extended models and corresponding quantities of interest. We demonstrate some of these properties by  deriving robust uncertainty quantification bounds for phase diagrams in statistical mechanics models. 
\end{abstract}
\bigskip
\end{frontmatter}

\section{Introduction}
Information Theory  provides  both mathematical  methods and practical 
computational tools to  construct probabilistic models in a principled 
manner, as well as the means  to assess their validity, \cite{cover}. One of the key 
mathematical  objects of  information theory is the concept of  
information metrics between probabilistic models. Such concepts of 
distance between models  are not always metrics in the strict 
mathematical sense, in which case they are called divergences, and  
include  the relative entropy, also known as the Kullback-Leibler 
divergence, the total variation and the  Hellinger metrics, the 
$\chi^2$ divergence, the F-divergence, and  the R$\acute{e}$nyi  divergence,
\cite{tsybakov2008introduction}. For example, the relative entropy 
between two probability distributions 
$P=P(x)$ and $Q=Q(x)$  on  $\R^N$ is defined as 
\begin{equation}
\label{relent:intro}
R(Q \mid \mid P)
=\int_{\R^N} \log\left(\frac{{Q(x)}}{{P(x)}}\right)\, Q(x) dx\, ,
\end{equation}
when the integral exists. The relative entropy is not a metric but it is a divergence, that is  
it satisfies the properties 
(i) $R(Q\mid\mid P)\geq0$,  (ii) $R(Q\mid\mid P)=0$ if and only  if
$P=Q$ a.e. 

We may for example think of the  model $Q$ as an approximation, or a 
surrogate model for another complicated and possibly inaccessible 
model $P$; alternatively we may consider the model $Q$ as a misspecification  of the true model $P$.  When measuring  model discrepancy  between the two models 
$P$ and $Q$, tractability depends critically on the type  of distance 
used between models.  In that respect, the relative entropy has very 
convenient 
analytic and computational properties, in particular regarding  
to the scaling properties of the system size $N$ which could represent 
space and/or time.  Obtaining bounds which are valid for high dimensional ($N \gg 1$) or spatially 
extended systems and/or long time regimes is the main topic of the 
paper and we will discuss  these properties in depth in the 
upcoming sections. 

Information metrics provides systematic and practical tools for 
building  approximate statistical models of reduced complexity through  
variational inference methods \cite{MacKay:2003,bishop,stuartKL}  for  
machine learning \cite{Wainwright:2008,Blei:2013,bishop} and  
coarse-graining  of complex  systems \cite{Shell2008,shell2010,NoidReview:2011,EspanolZuniga2011,Bilionis:2012,zabaras2013,PK2013,Foley2015,KP2016}. 
Variational inference relies on optimization problems such as 
\begin{equation}
\label{variational:intro}
\min_{Q\in \mathcal{Q}} R(P\mid\mid Q)\,  \quad\mbox{or}\quad   \min_{Q \in \mathcal{Q}} R(Q \mid\mid P)
\end{equation}
where $\mathcal{Q}$ is a class of simpler, computationally more tractable 
probability models than $P$. Subsequently, the optimal solution $Q^*$ 
of \eqref{variational:intro} replaces $P$ for estimation, simulation 
and  prediction purposes. The choice of order in  $P$ and $Q$ in 
\eqref{variational:intro} can be significant and depends on 
implementation methods, availability of data and the specifics of each 
application, e.g. \cite{MacKay:2003,bishop,PK2013,stuartKL}. In 
the case of coarse-graining the class of coarse-grained 
models $\mathcal{Q}$ will  also have fewer degrees of freedom than the 
model $P$, and an additional projection operator is  needed  in the 
variational principle \eqref{variational:intro}, see for instance 
\cite{Shell2008, KP2016}. In addition, information metrics provide  
fidelity measures in  model reduction, \cite{KV2003,majdaabramov2005, KT2006, KPRT1, majdageshgfidelity,KKP2013,majdatong:info_barriers}, sensitivity metrics  for   uncertainty quantification,  \cite{liu2006sens,ludtke2008sens,Komorowski,majdageshgsens,pantazis2013relative,lam:robustsens} and   discrimination criteria in model selection 
\cite{burnhammodelselec,kitagawa2008modelselec}. 
For instance, for the sensitivity analysis of parametrized probabilistic models $P^\theta=P^\theta(X)$, $\theta \in \Theta$ 
the relative entropy $R(P^\theta \mid \mid P^{\theta+\epsilon})$ 
measures the  loss of information due to an error in parameters in 
the direction of the vector $\epsilon \in \Theta$.  Different 
directions in parameter space provide a ranking of the sensitivities. 
Furthermore, when $\epsilon \ll 1$ we can also consider the quadratic  
approximation $R(P^\theta \mid \mid P^{\theta+\epsilon})=\epsilon 
\FISHER{P^\theta} \epsilon^\top+ O(|\epsilon|^3)$ where 
$\FISHER{P^\theta}$ is the Fisher Information matrix, 
\cite{majdageshgsens,Komorowski,pantazis2013relative}.




It is natural and useful to approximate, perform model selection and/or
sensitivity analysis in terms of information theoretical metrics between 
probability distributions.  However, one is often interested in 
assessing model approximation, fidelity  or sensitivity  on  
concrete quantities of interest and/or  statistical estimators. 
More specifically, suppose $P$ and $Q$ are  two probability measures 
and let $f=f(X)$ be some quantity of interest  or statistical 
estimator.  In variational inference one takes $Q=Q^*$ to be 
the solution of the optimization problem \eqref{variational:intro}, 
while in the context of sensitivity analysis we set $P=P^\theta$ and 
$Q=P^{\theta+\epsilon}$.
%
We then measure the discrepancy  between  models $P$ and $Q$ 
with respect to  the Quantity of Interest (QoI) $f$ by considering
\begin{equation}
\label{observables:intro}
 E_Q(f) -  E_P(f)\, .
\end{equation}
Our main mathematical goal is to understand how to transfer 
quantitative results on information metrics into bounds 
for quantities of interest in  \eqref{observables:intro}. 
In a statistics context, $f$ could be an unbiased  statistical 
estimator for model $P$  and thus \eqref{observables:intro} is the 
estimator bias when using model $Q$ instead of $P$. 

In this direction, information inequalities can provide a method to relate quantities of 
interest \eqref{observables:intro} and information metrics \eqref{relent:intro}, a classical example being the Csiszar-Kullback-Pinsker (CKP) inequality, \cite{tsybakov2008introduction}:
\begin{equation}
\label{CKP:intro}
| E_Q(f) -  E_P(f)|
\leq ||f||_\infty \sqrt{2 R(Q\mid\mid P)}
\end{equation}
where $||f||_\infty=\sup_{X \in \R^N}|f(X)|$. In other words relative entropy  controls how 
large the model discrepancy \eqref{observables:intro} can become  for 
the quantity of interest $f$. More such inequalities involving other probability metrics such as 
Hellinger distance, $\chi^2$ and R$\acute{e}$nyi  divergences are discussed in the 
subsequent sections. 

In view of \eqref{CKP:intro} and other such inequalities, a natural 
question is whether these are sufficient to assess the fidelity 
of complex systems models.  In particular complex  systems such as 
molecular  or multi-scale models are typically high dimensional in  the  
degrees of freedom and/or  often require controlled fidelity (in 
approximation, uncertainty quantification, etc) at long time regimes;  
for instance, in building coarse-grained models for efficient  and 
reliable molecular simulation.  
Such an example arises when we are comparing two statistical mechanics systems 
determined by  Hamiltonians $H_N$ and $\bar{H}_N$ describing say $N$ particles with positions $X=(x_1,...,x_N)$. 
The associated canonical Gibbs measures are given by  
\begin{equation}
\label{gibbs:intro}
P_N(X)dX=Z_N^{-1}e^{-H_N(X)}dX\, \quad 
\mbox{and} \quad Q_N(X)dX=\bar{Z}_N^{-1}e^{-\bar{H}_N(X)}dX\, .
\end{equation}
where $Z_N$ and ${\bar Z}_N$ are normalizations (known as partition 
functions) that ensure the measures \eqref{gibbs:intro} are 
probabilities. Example \eqref{gibbs:intro} is a ubiquitous one, given 
the importance  of Gibbs measures  in disparate fields ranging from 
statistical mechanics and molecular simulation,  pattern recognition 
and image analysis, to machine and statistical learning, 
\cite{simon2014statistical,MacKay:2003,bishop}. In the case of 
\eqref{gibbs:intro}, the relative entropy \eqref{relent:intro} 
readily yields, 
\begin{equation}
\label{scaling:intro}
R(Q_N\mid\mid P_N)= E_{Q_N}(H_N-\bar{H}_N) +\log Z_N-\log\bar{Z}_N\, .
\end{equation}
It is a well known result in classical statistical mechanics 
\cite{simon2014statistical}, that under very general assumptions on 
$H_N$, both terms in the right hand side of \eqref{scaling:intro} scale 
like $O(N)$ for $N\gg 1$, therefore we have that
\begin{equation}
\label{scaling:intro:2}
R(Q_N\mid\mid P_N)=O(N)\, .
\end{equation}
Comparing to  \eqref{CKP:intro}, we immediately realize that the upper 
bound grows  with the system size $N$,  at least for 
nontrivial quantities of interest $f$ and therefore the CKP inequality 
\eqref{CKP:intro} yields no information on  model discrepancy for  
quantities of interest in \eqref{observables:intro}. In Section 2 we 
show that  other known information inequalities 
involving other divergences are also  
inappropriate for large systems in the sense that they do not provide 
useful information for quantities of interest: they either blow up like 
\eqref{CKP:intro} or lose their selectivity, in the $N \gg 1$ limit.  
Furthermore, in Section 2 we also show that similar issues arise for 
time dependent stochastic Markovian models at long time regimes, $T\gg 
1$. 

In our  main result we address  these issues  by using the recent information inequalities of 
\cite{dupuis2015path} which in turn relied on earlier upper bounds in 
\cite{dupuis2011uq}. In these  inequalities,
the discrepancy in quantities of interest 
\eqref{observables:intro}
is bounded as follows: 
\begin{equation}\label{UQII:intro}
\Xi_{-}(Q \mid \mid P; f)
\leq E_{Q}(f) - E_{P}(f) \leq
\Xi_+(Q \mid \mid P; f)\, .
\end{equation}
where 
\begin{equation}
\label{eq:goaldivplus:intro}
\Xi_+(Q \mid \mid P; f)= \inf_{c>0}\left\{\frac{1}{c} \log E_P\left( e^{c (f - E_P(f))}\right) + \frac{1}{c} R(Q \mid \mid P)\right\}\, .
\end{equation}
with a similar formula for $\Xi_-(Q \mid \mid P; f)$. 
The roles of $P$ and $Q$ in \eqref{UQII:intro} can be reversed as in \eqref{variational:intro}, depending on the context and the challenges  of the specific problem, as well as on how easy it is to compute or bound  the terms involved in \eqref{eq:goaldivplus:intro}; we discuss specific examples in Section 6.

The quantities  $\Xi_{\pm}(Q \mid \mid P; f)$ are referred to as a ``goal-oriented divergence'', \cite{dupuis2015path}, 
because they have  the properties 
of a divergence both in  probabilities $P$ and $Q$  and the quantity of interest  $f$. More precisely, 
$\Xi_{+}(Q\mid \mid P; f)\ge 0$, (resp. $\Xi_{-}(Q \mid \mid P; f) \le 0$) and  $\Xi_{\pm}(Q\mid \mid P; f)=0$ if and only if $P=Q$ a.s. 
or $f$ is constant $P$-a.s. 

The bounds \eqref{UQII:intro} turn out to be  robust, i.e. the bounds 
are attained when considering  the set of all models $Q$ with  a 
specified uncertainty threshold $\eta$ within the model $P$ given by  
the distance $R(Q \mid \mid P) \le \eta$;  we refer to   
\cite{dupuis2011uq}, while related robustness results  can  be also 
found in \cite{glasserman2014}.
The parameter $c$ in the variational representation \eqref{eq:goaldivplus:intro} 
controls the degree of robustness with respect to the model uncertainty captured by $R(Q\mid \mid P)$. In a control or optimization context these bounds are also related to  H$^\infty$ control, \cite{dupuis2000}.
Finally,  $\Xi_{\pm}(Q\|P; f)$  admits an asymptotic expansion in  
relative entropy, \cite{dupuis2015path}:
\begin{equation}
\label{linearized:intro}
\Xi_{\pm}(Q\mid\mid P; f)=\pm\sqrt{{Var}_P[f]}\sqrt{2 R(Q\mid\mid P)} + O(R(Q\mid\mid P))\, ,
\end{equation}
which  captures the  aforementioned  divergence properties, at least to 
leading order.

In  this paper we  demonstrate  that the bounds 
\eqref{UQII:intro}   
scale correctly with the system size $N$ and provide ``scalable" uncertainty quantification bounds for large classes of QoIs. 
We can get a first indication that this is the case by considering 
the  leading term in the expansion \eqref{linearized:intro}. 
On one hand, typically  for high dimensional systems we have  $R(Q\mid\mid P)=O(N)$, see for instance \eqref{scaling:intro:2}; but on the other hand for common 
quantities of interest, e.g. in molecular systems non-extensive QoIs such as density, 
average velocity, magnetization or 
specific energy,   we expect to have 
\begin{equation}
\label{var:intro}
Var_P(f)=O(1/N)\, .
\end{equation}
Such QoIs also include many statistical estimators e.g. those 
with asymptotically normal behavior such as sample means or  maximum 
likelihood estimators, \cite{kitagawa2008modelselec}.
Combining estimates \eqref{scaling:intro:2} and 
\eqref{var:intro}, we see that, at least to leading order, the bounds in \eqref{UQII:intro} scale as 
$$\Xi_{\pm}(Q \mid \mid P; f) \approx  O(1)\, ,$$  in sharp contrast
to the CKP inequality \eqref{CKP:intro}. Using tools from 
statistical mechanics we show that this scaling 
holds not only  for the leading-order term but for the goal oriented divergences  $\Xi_{\pm}(Q\mid\mid P; f)$ themselves,  
for extended systems such as Ising-type model in the  thermodynamic 
limit. These results are presented in Sections 3 and 4. 
Furthermore, in \cite{dupuis2015path}
it is also shown that such information inequalities  can be used to 
address model error  for time dependent problems at long time regimes. 
In particular our results extend to path-space observables, e.g.,
ergodic averages, correlations, etc, where the role of relative entropy 
is played by the  relative entropy rate (RER) defined as the  relative 
entropy per unit time. We revisit the latter point here and connect it 
to nonlinear response calculations for stochastic dynamics  in 
statistical mechanics.

Overall, the scalability of \eqref{UQII:intro} allows us to address 
three challenges which are not readily covered by standard numerical 
(error) analysis, statistics  or statistical mechanics calculations: 
(a) obtain uncertainty quantification  (UQ) bounds for quantities of 
interest  in the large degree of freedom  limit $N\gg 1$  and/or at 
long time regimes $T\gg 1$, (b) estimate the impact of  large model 
perturbations, going beyond  error expansion methods and 
providing  nonlinear response bounds   in the the  statistical 
mechanics sense,  and (c) address model-form uncertainty, i.e. 
comparing different extended models and corresponding  quantities of 
interest (QoIs).

We demonstrate all three capabilities in  deriving robust uncertainty 
quantification bounds for phase diagrams in statistical mechanics 
models. Phase diagrams are calculations of QoIs as functions of continuously varying model parameters, e.g. temperature, external forcing, etc. 
Here we consider a given model $P$ and desire to calculate   uncertainty bounds for its phase diagram, when the model $P$ is replaced by a different  model $Q$. 
We note that phase diagrams are typically computed    in the the thermodynamic limit $N \to \infty$ and in 
the case of steady states in the long time regime $T \to \infty$; thus, in order to obtain uncertainty bounds for the phase diagram of the model $P$,  we necessarily will require  scalable bounds such as  \eqref{UQII:intro}; similarly,  we need such scalable bounds to  address  any related  UQ or sensitivity analysis question for molecular or any other high-dimensional probabilistic  model. 
To illustrate the potential  of our methods,  we consider fairly large  
parameter discrepancies between models $P$ and $Q$  of the order of 
$50\%$ or more, see for instance Figure \ref{Fig:Figure_intro}(a). 
We  also compare phase diagrams corresponding not just to different 
parameter choices but to entirely different Gibbs models 
\eqref{gibbs:intro}, where $Q$  is a true microscopic model and $P$ is 
for instance some type of mean field approximation, see Figure \ref{Fig:Figure_intro} (b).  These and several  other test-bed examples are discussed in detail  in Section 5.  
It should be noted that the bounds in Figure \ref{Fig:Figure_intro}
are very tight  once we compare to  the real approximated 
microscopic model, see the figures in Section 5.  

\begin{figure}[H]
	\centering
	\begin{subfigure}[b]{0.45\textwidth}
		\centering
		\includegraphics[width=\linewidth]{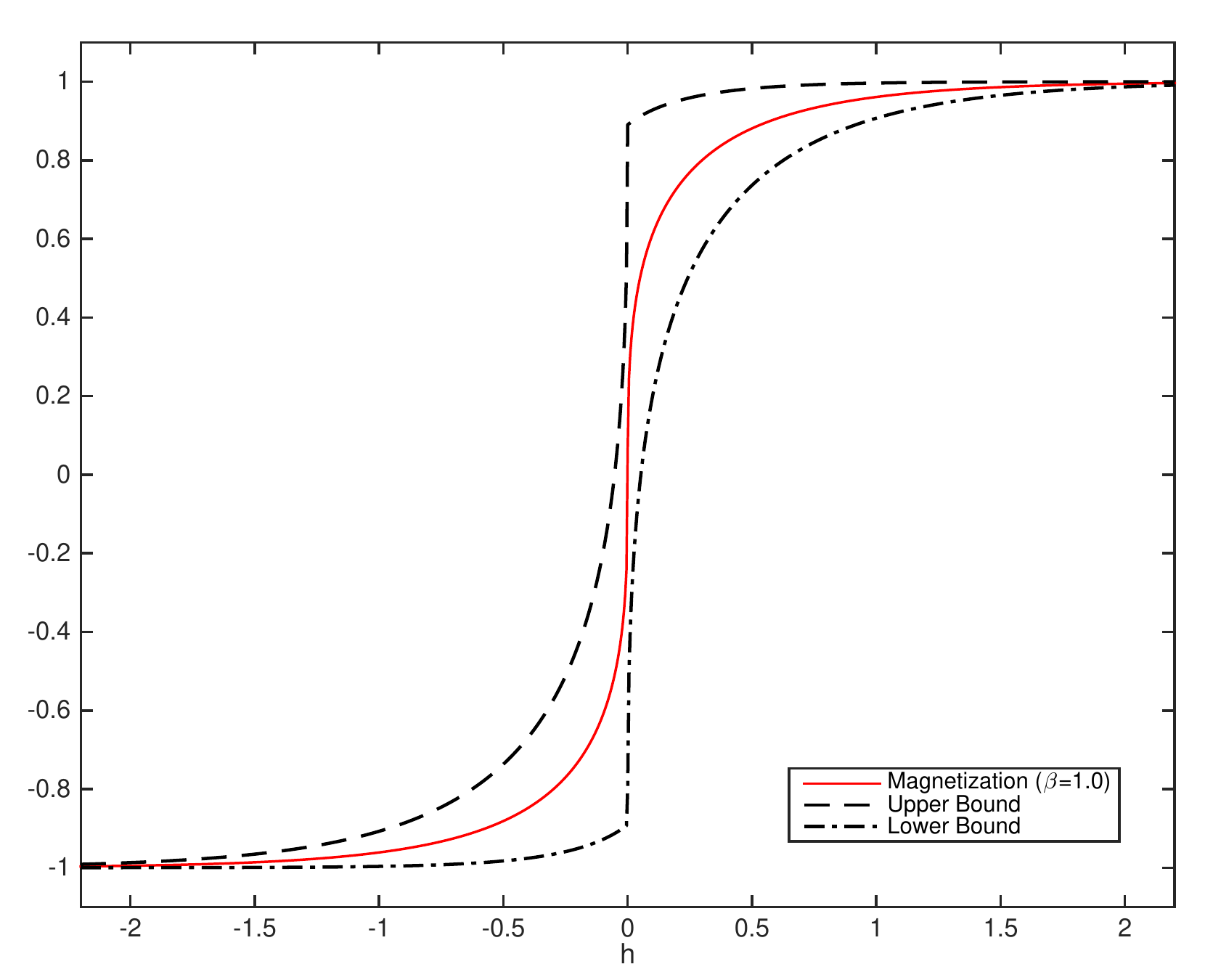}
		\caption{}\label{Fig:meanfield_intro}
	\end{subfigure}%
	\begin{subfigure}[b]{0.45\textwidth}
		\centering
		\includegraphics[width=\linewidth]{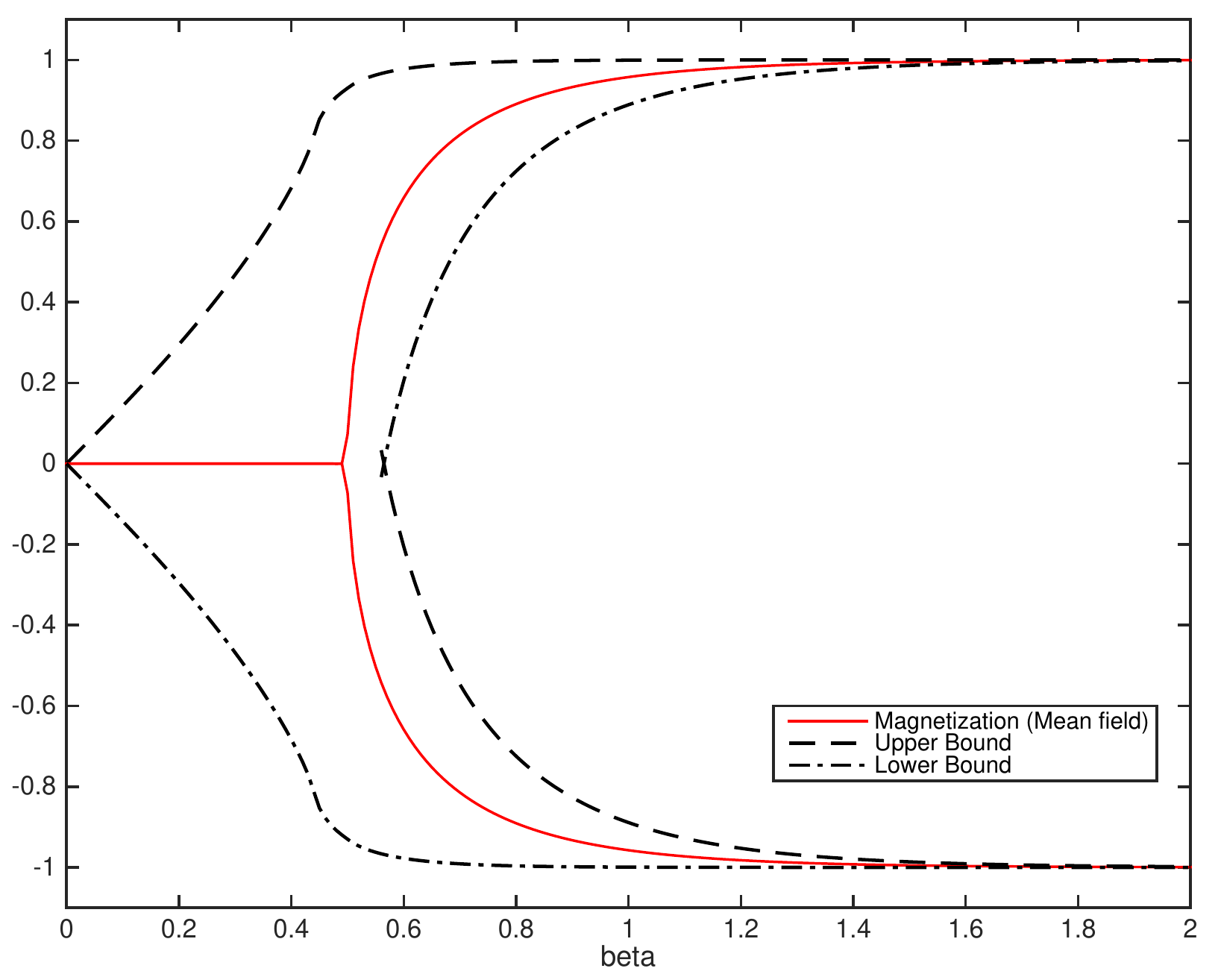}
		\caption{}
		\label{Fig:2dIsing_intro}
	\end{subfigure}
	\caption{\subref{Fig:meanfield_intro} The red solid line is the magnetization of 1-d mean field with $\beta=1$, which is the baseline; 
The black dashed/dash-dot lines is the goal-oriented divergence upper/lower bound of magnetization of the mean field with $\beta=1.6$.  \subref{Fig:2dIsing_intro} The red solid line is the spontaneous magnetization of 2-d mean field Ising model for $\beta=1$, which is the baseline; 
The black dashed/dash-dot lines is the goal-oriented divergence upper/lower bound of the  magnetization of the nearest neighbor Ising model with $\beta=1.6$.}
	\label{Fig:Figure_intro}
\end{figure}
This paper is organized as follows. In Section 2 we discuss classical information inequalities for QoIs and demonstrate that they do not scale with system size or with long time dynamics. We show these results  by  considering counterexamples such as sequences of independent identically distributed random variables and Markov chains. In Section 3 we revisit the concept of goal oriented divergence introduced earlier in \cite{dupuis2015path} and  show  that it provides  scalable and discriminating information bounds for QoIs. 
In Section 4 we discuss how these  results extend to path-space observables, e.g.,
ergodic averages, autocorrelations, etc, where the role of relative entropy 
is now played by the  relative entropy rate (RER)  and connect  
to nonlinear response calculations for stochastic dynamics  in 
statistical mechanics.
In Section 5 we show how these new information inequalities   transfer to Gibbs measures, implying nonlinear response UQ bounds,  and how they relate to classical results for thermodynamic limits in statistical mechanics.
Finally in Section 6 we apply our methods and the scalability of the UQ bounds to assess  model and parametric uncertainty of phase diagrams in molecular systems. We demonstrate the methods for Ising models, although the perspective is generally applicable.

\section{Poor scaling properties of the classical inequalities for 
probability metrics}

In this section we discuss several classical information inequalities
and demonstrate they scale  poorly with the size of the system especially when applying the 
inequalities to ergodic averages. We make these points by considering  simple examples such as independent, identically distributed (IID)
random variables, as well as  Markov sequences  and the corresponding statistical estimators.

Suppose $P$ and $Q$ be two probability measures on some measure space 
$(\mathcal{X},\mathcal{A})$ and let $f : \mathcal{X} \to \mathbb{R}$ be 
some quantity of interest (QoI).  
Our goal is   to consider the discrepancy  between  models $P$ and $Q$ with respect to  the quantity of interest $f$,
\begin{equation}
\label{observables:sec2}
 E_Q(f) - E_P(f)\, .
\end{equation}
Our primary mathematical challenge here is to understand what  results on  information metrics between probability measures $P$ and $Q$ imply for quantities of interest in  \eqref{observables:sec2}. We first discuss several concepts of information metrics, including divergences and probability distances.

\subsection{Information distances and divergences} 
To keep the notation as simple as possible we will  assume 
henceforth that $P$ and $Q$ are mutually absolutely continuous and this will cover 
all the examples considered here. (Much of what we discuss would extend to general 
measures by considering a measure dominating $P$ and $Q$, e.g. $\frac{1}{2}(P+Q)$.
For the same reasons of simplicity in presentation, we assume that all integrals below exist and are finite.

\medskip 
\noindent 
{\bf Total Variation} \cite{tsybakov2008introduction}: The total variation distance 
between $P$ and $Q$ is defined  by
\begin{eqnarray}
TV(Q,P)= \sup_{A\in \mathcal{A}} \left| Q(A)-P(A) \right| = \frac{1}{2} \int \left|1 - \frac{dP}{dQ}\right| dQ .
\end{eqnarray}
Bounds on $TV(Q,P)$ provide bounds on $|E_Q(f) -E_P(f)|$ since we have 
\begin{equation}
TV(Q,P)\,=\, \frac{1}{2} \sup_{\|f\|_\infty=1} 
\left| E_Q(f) - E_P(f)\right| .
\end{equation}

\smallskip
\noindent
{\bf Relative entropy} \cite{tsybakov2008introduction}:
The Kullback-Leibler divergence, or relative entropy, of $P$ with respect to $Q$ is defined   by
\begin{equation}
R(Q\mid\mid P)= \int \log\left(\frac{dQ}{dP}\right)dQ .
\end{equation}

\smallskip
\noindent
{\bf Relative R$\acute{e}$nyi  entropy} \cite{liese2007convex}: For $\alpha >0$, $\alpha \not=1$,
the relative R$\acute{e}$nyi  entropy (or divergence) of order $\alpha$ of $P$ with 
respect to $Q$ is defined  by
\begin{equation}
D_\alpha(Q\mid\mid P)= \frac{1}{\alpha-1} \int \left(\frac{dP}{dQ}\right)^{1-\alpha} dQ = \frac{1}{\alpha-1} \int \left(\frac{dQ}{dP}\right)^{\alpha} dP .
\end{equation}

\smallskip
\noindent
{\bf $\chi^2$ divergence} \cite{tsybakov2008introduction}: The $\chi ^2$ divergence between $P$ and $Q$ 
is defined by:
\begin{equation}
\chi^2(Q\mid\mid P)= \int \left(\frac{dQ}{dP}-1\right)^2 dP
\end{equation}

\smallskip
\noindent
{\bf Hellinger distance} \cite{tsybakov2008introduction}: The Hellinger distance between $P$ and $Q$ 
is defined by:
\begin{equation}
H(Q,P)= \left( \int\left(1-\sqrt{\frac{dP}{dQ}}\right)^2 dQ\right)^{1/2}
\end{equation}

The total variation and Hellinger distances define proper distances while 
all the other quantities are merely divergences (i.e., they are non-negative and 
vanish if and only if $P=Q$).  The R$\acute{e}$nyi  divergence of order $1/2$ 
is symmetric in $P$ and $Q$ and is related to the Hellinger distance by 
\[D_{1/2}(Q \mid \mid P) = - 2 \log\left( 1 - \frac{1}{2} H^2(Q,P)\right)\,.\] 
Similarly  the R$\acute{e}$nyi  divergence of order $2$ is related to the $\chi ^2$ divergence
by \[D_{2}(Q \mid \mid P) = \log\left( 1+   \chi^2(Q\mid\mid P)\right)\,.\]
In addition the R$\acute{e}$nyi  divergence of order $\alpha$ is an nondecreasing 
function $\alpha$ \cite{van2014renyi} and we have 
\[
\lim_{\alpha \to 1}D_\alpha(Q\mid\mid P) = R(Q\mid\mid P).
\]
and thus it thus natural to set $D_1(Q\mid\mid P) = R(Q\mid\mid P)$. 
Using the inequality $\log(t)\le t-1$ we then obtain the chain of inequalities, \cite{tsybakov2008introduction}
\begin{equation}\label{eq:ineqs}
H^2(Q,P) \le D_{1/2}(Q \mid \mid P) \le R(Q\mid\mid P) \le D_{2}(Q \mid \mid P) \le \chi^2( Q \mid \mid P) \,.
\end{equation}






%
%
%
%
%

%

\subsection{Some classical information inequalities for QoIs}

We recall a number of classical information-theoretic bounds 
which use probability distances ot divergences to control expected values of QoIs
(also referred to as observables). Because 
$$|E_Q(f) -  E_P(f)| \le 2 \|f\|_\infty TV(Q,P)$$ we can readily   
obtain bounds on QoIs from relationships between $TV(Q,P)$ and other 
divergences. It is well-known and easy to prove that 
$TV(Q,P) \le H(Q,P)$ but we will use here the slightly sharper bound (Le Cam's inequality) \cite{tsybakov2008introduction} given by \[
TV(Q,P) \le H(Q,P)\sqrt{1 -\frac{1}{4}H^2(Q,P)}\]
which implies 

\medskip
\noindent{\bf Le Cam}\cite{tsybakov2008introduction}: 
\begin{equation}
\left| E_{Q}(f)-E_{P}(f) \right| \le 2 \|f\|_\infty  
H(Q,P)\sqrt{1 - \frac{1}{4}H^2(Q ,P)}.
\label{eq:hellinger}
\end{equation}
\medskip

From inequality \eqref{eq:ineqs} and $TV(Q,P)\le H(Q,P)$ 
we obtain immediately bounds on $TV(Q,P)$ by 
$\sqrt{D_\alpha(Q\mid\mid P)}$ but the constants are not 
optimal. The following generalized Pinsker inequality (with optimal 
constants)  was proved in \cite{gilar2010} and holds for 
$0< \alpha \le 1$ 
\[TV(Q,P) \le \sqrt{ \frac{1}{2\alpha} D_\alpha(Q\mid\mid P)},
\]
and leads to 

\medskip
\noindent{\bf Csiszar-Kullback-Pinsker (CKP)} \cite{tsybakov2008introduction}:
\begin{equation}
\left| E_{Q}(f)-E_{P} (f)\right|  \le  \|f\| _{\infty}\sqrt{2R(Q \mid \mid P)}. 
\label{eq:REn}
\end{equation}

\smallskip
\noindent
{\bf Generalized Pinsker}: see { \cite{van2014renyi}}:  
For $0 < \alpha \le 1$ 
\begin{equation}
\left| E_{Q}(f)-E_{P} (f)\right|  \le  \|f\|_{\infty} \sqrt{\frac{2}{\alpha}D_\alpha (Q \mid \mid P)}. 
\label{eq:Renyi }
\end{equation}
\medskip 

It is known that the CKP  inequality is sharp only of $P$ and $Q$ 
are close. In particular the total variation norm is always 
less than $1$ while the relative entropy can be very large. 
There is a complementary bound to the CKP inequality which is based 
on a result by Scheff\'e \cite{tsybakov2008introduction}

\medskip
\noindent 
{\bf Scheff\'e}: 
\begin{equation}
\mid E_{Q}(f)-E_{P} (f)\mid  \le   \| f \|_{\infty} \left(2 - e^{- R(Q\mid\mid P)}\right).
\label{eq:CoREn}
\end{equation}

By \eqref{eq:ineqs} we have $R(P\mid\mid Q) \le \chi^2(Q\mid\mid P)$ 
and thus we can also obtain a bound in terms  of the $\chi^2$ divergence 
and $\|f\|_\infty$. However,  we can obtain a better bound  which involves the 
variance of  $f$ by using Cauchy-Schwartz inequality. 

\medskip
\noindent{\bf Chapman-Robbins} \cite{lehmanncasella2003pointestimation}:
\begin{equation}
\left| E_{Q}(f)-E_{P} (f)\right| \le\sqrt{Var_{P}(f)}\sqrt{\chi^2(Q \mid \mid P)}.
\label{eq:chisquare}
\end{equation}
\medskip

\medskip
\noindent{\bf Hellinger-based inequalities:}

Recently a bound using the Hellinger distance and the $L^2$ norm was 
derived in \cite{dashti2013bayesian}:  
\begin{eqnarray*}
\left| E_{Q}(f)-E_{P}(f) \right|
& \le & \sqrt{2}H(Q,P) \sqrt{(E_P (f^2) +E_Q (f^2))}.
\label{eq:hellinger_dashti} 
\end{eqnarray*}
As we show in  Section \ref{section:Hellinger} this bound can be further optimized  by using a control variates argument.
Note that the left hand side is unchanged by replacing $f$ by 
$f - \frac{1}{2}(E_P(f)+E_Q(f))$ and this yields the improved  bound

\begin{eqnarray}
 \left| E_{Q}(f)-E_{P}(f)\right|
& \le & \sqrt{2}H(Q,P) \sqrt{ Var_P(f) + Var_P(f) + \frac{1}{2}(E_Q(f) - E_P(f))^2}. 
\label{eq:hellinger_MK}
\end{eqnarray}

\subsection{Scaling properties for IID sequences}\label{section:iid_1}

We make here some simple, yet useful  observations, on how the inequalities discussed in the previous Section
scale with system size for IID sequences. We consider the product measure space 
$\mathcal{X}_N=\mathcal{X}\times \cdots \times \mathcal{X}$ equipped with the 
product $\sigma$-algebra $\mathcal{A}_N$ and we denote by $P_N=P \times \cdots 
\times P$ the product  measures on $(\mathcal{X}_N,\mathcal{A}_N)$ 
whose all marginals  are equal to $P$ and we define $Q_N$ similarly. From a statistics perspective, this is also the setting where  sequences of $N$ independent samples are generated by the models $P$ and $Q$ respectively.

We will concentrate on QoIs which are observables which have the form of ergodic averages or of statistical estimators. The challenge would be to assess based on information inequalities the impact on the QoIs. Next, we  consider the simplest such case  of the  sample mean.
For any measurable $g: \mathcal{X} \to \mathbb{R}$ we consider the observable  $f_N:\mathcal{X}_N \to \mathbb{R}$ given by 
\begin{equation}
\label{mean:sec2}
f_N(\sigma_1, \cdots \sigma_N) = \frac{1}{N}\sum_{j=1}^N g(\sigma_j)\,,
\end{equation}
This quantity is also the sample average of the data set 
${\mathcal D}=\{\sigma_1, \cdots \sigma_N \}$.
We also note that  
\[ \|f_N\|_\infty = \|g\|_\infty \,,\quad \quad  E_{P_N}(f_N) = E_{P}(g)\,,\quad \quad  \quad \quad Var_{P_N}(f_N) = \frac{1}{N}Var_P(g)
\]

To understand how the various inequalities scale with the system size $N$
we need to understand how the information distances and divergences themselves 
scale with $N$. For IID random variables the results are collected in the following Lemma.

\begin{lemma}\label{lem:scaling_iid} For two product measures $P_N$ and $Q_N$ with marginals $P$ and $Q$ we have 
\begin{eqnarray}
&\textrm{\bf Kullback-Leibler:}&R(Q_N\mid\mid P_N)= N R(Q\mid\mid P)
\label{eq:re_N} \nonumber \\
&\textrm{\bf R\'enyi :}&D_\alpha (Q_N\mid\mid P_N) =N D_\alpha(Q \mid \mid P)
\label{eq:renyi_N}\nonumber \\
&\textrm{\bf Chi-squared:} &\chi^2(Q_N \mid \mid P_N) =\left(1+ \chi^2(Q \mid \mid P)\right)^N-1 \label{eq:chisquare_N} \nonumber \\
&\textrm{\bf Hellinger:}&H(Q_N \mid \mid P_N) =\sqrt{2-2\left(1-\frac{H^2(Q ,P)}{2}\right)^N} \label{eq:hellinger_N}
\end{eqnarray}
\end{lemma}

\proof See 
\ref{app:iid_markov}.

\medskip
Combining the result in Lemma \ref{lem:scaling_iid} with the information 
bounds in the previous Section we obtain a series of bounds for ergodic 
averages which all suffer from serious defects. Some grow to infinity for $N\gg 1$ while others converge to a trivial bound that is not discriminating, namely provide no new information on the difference of the QoIs $E_{Q_N}(f_N)-E_{P_N} (f_N)$.
More precisely we obtain the following bounds:

\medskip
\noindent 
{\bf Csiszar-Kullback-Pinsker (CKP) for IID}: 
\begin{equation}
\left| E_{Q_N}(f_N)-E_{P_N} (f_N)\right|  \le   \| g \|_{\infty}\sqrt{2 N R(Q \mid \mid P)} = O(\sqrt{N}).
\label{eq:REn_N_CKP}
\end{equation}
\medskip
\noindent 
{\bf Generalized Pinsker for IID}:  For $0 < \alpha <1$ we have
\begin{equation}
\left|  E_{Q_N}(f_N)-E_{P_N} (f_N)\right|  \le   \| g\| _{\infty}\sqrt{\frac{2N}{\alpha} D_\alpha (Q \mid \mid P)} = O(\sqrt{N}). 
\label{eq:Renyi n}
\end{equation}

\medskip
\noindent 
{\bf Scheff\'e for IID}: 
\begin{equation}
\left| E_{Q_N}(f_N)-E_{P_N} (f_N)\right| \le   \| g \|_{\infty} \left(2 - e^{-N R(Q\mid\mid P)}\right)  =O(1).
\label{eq:REn_N_Scheffe}
\end{equation}

\medskip
\noindent 
{\bf Chapman-Robbins for IID}: We  have
\begin{equation}
\left| E_{Q_N}(f_N)-E_{P_N}(f_N)\right| \le\sqrt{\frac{1}{N}Var_{P}(g)}\sqrt{ \left(1+ \chi^2(Q \mid \mid P)\right)^N-1} = O\left(\frac{{\sqrt{e^N}}}{\sqrt{N}}\right). \label{eq:chisquaren}
\end{equation}

\medskip
\noindent 
{\bf Le Cam for IID}:
\begin{eqnarray}
\left| E_{Q_N}(f_N)-E_{P_N}(f_N)  \right|  & \le &  2 \|g\|_\infty  
\sqrt{ 2-2 \left( 1-\frac{H^2(Q ,P)}{2}\right)^N} 
\sqrt{ \frac{1}{2} + \frac{1}{2} \left( 1-\frac{H^2(Q ,P)}{2}\right)^N}
\nonumber \nonumber \\
&=&  O(1).
\label{eq:hellingern_LeCam}
\end{eqnarray}

\medskip
\noindent 
{\bf Hellinger for IID}:  
\begin{align}
 \left| E_{Q_N}(f_N)-E_{P_N}(f_N) \right|
& \le & \sqrt{2}\sqrt{2-2\left(1-\frac{H^2(Q ,P)}{2}\right)^N}\sqrt{ \frac{Var_P(g)}{N} + \frac{Var_Q(g)}{N} + \frac{1}{2} (E_P(g) - E_Q(g))^2}.
\label{eq:hellingern}
\end{align}

Every single bound fails to capture the behavior of ergodic averages. 
Note that the left-hand sides are all of order $1$ and indeed should be small 
of $P$ and $Q$ are sufficiently close to each other.  The CKP, 
generalized 
Pinsker and Chapman-Robbins bounds all diverge as $N \to \infty$ and thus 
completely fail.  The Le Cam bound is of order $1$, but as $N \to \infty$ the 
bound converges to $2 \|f\|_\infty$ which is a trivial bound independent of $P$ and $Q$.   
The Scheff\'e likewise converges to constant. Finally the Dashti-Stuart bound
converges to the trivial statement that $1\le \sqrt{2}$.

\subsection{Scaling properties for Markov sequences}
\label{section:markov_1}

Next, we  consider the same questions as in the previous Section, however this time for correlated distributions. Let two Markov chains in a finite state space $S$ with transitions matrix $p(x,y)$ 
and $q(x,y)$ respectively.  We will assume that both Markov chains are irreducible  
with stationary  distributions $\mu_p$ and $\mu_q$ respectively.  In addition  
we assume that for any $x \in S$, the probability measure $p(x, \cdot)$ and 
$q(x, \cdot)$ are mutually absolutely continuous. We denote by $\nu_p(x)$ and $\nu_q(x)$ the initial distributions of the two Markov chains 
and then the probability distributions of the path  $(X_1, \cdots X_N)$ evolving under $p$ 
is given by 
\[
P_N(X_1, \cdots, X_N) = \nu_p(X_0)p(X_0,X_1) \cdots p(X_{N-1}, X_N)\,,
\]
and similarly for the distribution $Q_N$ under $q$.

If  we are interested in the long-time behavior of the system,
for example  we may be interested in computing or estimating expectations of the steady state or in our case model discrepancies such as 
\[
\left| E_{\mu_q}(g) -E_{\mu_p}(g)\right| 
\]
for some QoI (observable) $g : 
S \to \mathbb{R}$.  In general the steady state of a Markov chain  
is not known explicitly  or it is difficult to compute for large systems.  
However, if we consider ergodic observables such as 
\begin{equation}
\label{ergodic:obs:sec2}
f_N(X_1, \cdots X_N) =\frac{1}{N} \sum_{i=1}^N g(X_i)
\end{equation}
then,  by the ergodic theorem, we have, for 
any initial distribution $\nu_p(x)$  that 
\[
\lim_{N \to \infty} E_{P_N}(f_N) = E_{\mu_p}(g)\,.
\]and thus can estimate $|E_{\mu_q}(g)- E_{\mu_p}(g)|$  if we can control $|E_{Q_N}(f)-E_{P_N}(f)|$ for large $N$. 
After our computations with IID sequences in the previous Section, it is not surprising    that none of the standard information inequalities allow such 
control.  
Indeed the following lemma, along with the fact that the variance of ergodic observables such as \eqref{ergodic:obs:sec2} scales like 
$Var_{P_N}(f_N)=O(1/N)$ \cite{dupuis2015path}, readily imply 
 that the bounds for Markov measures 
scale  exactly as (poorly as) the IID case, derived at the end of Section \ref{section:iid_1}.

\begin{lemma}\label{lem:scaling_markov} Consider two irreducible Markov chains 
with transitions matrix $p$ and $q$.  Assume that the initial conditions  $\nu_p(x)$ and $\nu_q(x)$ 
are mutually absolutely continuous and that $p(x,\cdot)$ and $q(x,\cdot)$ are  mutually absolutely 
continuous for each $x$.  

\medskip
\noindent
{\bf Kullback-Leibler:} We have  
\[
\lim_{N \to \infty} 
\frac{1}{N} R(Q_N\mid\mid P_N)=r(q\mid\mid p):= \sum_{x,y} \mu_q(x) q(x,y)
\log \left( \frac{q(x,y)}{p(x,y)} \right)\,,
\]
and the limit is positive if and only if $p\not= q$.

\smallskip
\noindent
{\bf R\'enyi :} We have 
\[
\lim_{N \to \infty} \frac{1}{N} D_\alpha(Q_N\mid\mid P_N)= \frac{1}{\alpha -1}\log \rho(\alpha) \,,
\]
where $\rho(\alpha)$ is the maximal eigenvalue of the non-negative matrix 
with entries $q^{\alpha}(x,y) p^{1-\alpha}(x,y)$ and we have 
$\frac{1}{\alpha -1}\log \rho(\alpha) \ge 0$  with equality if and only if $p\not= q$.  

\smallskip
\noindent
{\bf Chi-squared:} We have  
\[
\lim_{N \to \infty} \frac{1}{N} \log( 1 + \chi^2(Q_N \mid \mid P_N)) = \log \rho(2)
\]
where $\rho(2)$ is the maximal eigenvalue of the matrix 
with entries $q^{2}(x,y) p^{-1}(x,y)$ and we have $\log\rho(2)\ge 0$ 
with equality if and only if $p = q$. 

\smallskip
\noindent {\bf Hellinger:} We have 
\[
\lim_{N \to \infty} H(Q_N \mid \mid P_N) = \sqrt{2}.
\]
if $p \not= q$ and $0$ if $p=q$.
\end{lemma}

\proof See Appendix \ref{app:iid_markov}.

\section{A  divergence with good scaling properties}

\subsection{Goal Oriented Divergence}
\label{section:goal-oriented}

In this Section we will first discuss the goal-oriented divergence which was  introduced 
by  \cite{dupuis2015path}, following the work in \cite{dupuis2011uq}.  
Subsequently in Sections \ref{section:iid_2}  and \ref{section:markov_2} and Section \ref{section:Gibbs}  we will demonstrate that this new divergence provides  bounds on the model discrepancy $E_Q(f) - E_P(f)$  between models $P$ and $Q$  which scale correctly with their 
system size,  provided the QoI $f$ has  the form of an ergodic average or  a statistical estimator.

Given an observable $f: \mathcal{X} \to \mathbb{R}$  we introduce  the cumulant generating function of $f$  
\begin{eqnarray}
\Lambda_{P,f} (c) = \log E_{P} (e^{cf}) .
\end{eqnarray}
We will assume $f$ is such that  $\Lambda_{P,f} (c)$ is finite in a neighborhood $(-c_0,c_0)$ of  the origin. 
For example if  $f$ is bounded then we can take $c_0= \infty$.  Under 
this assumption $f$ has  finite moments of any order and we will often 
use   the cumulant generating function of a mean  $0$ observable 
\begin{eqnarray}
\widetilde{\Lambda} _{P,f}(c) = \log E_{P} (e^{c(f-E_P(f))})  = \Lambda_{P,f} (c) - cE_P(f).
\end{eqnarray} 
The following bound is proved in \cite{dupuis2015path} and will play a 
fundamental role in the rest of the paper. 

\bigskip
\noindent
{\bf Goal-oriented divergence UQ bound:} If $Q$ is absolutely 
continuous with respect to $P$ and ${\Lambda} _{P,f}(c)$ is finite 
in a neighborhood of the origin, then 
\begin{equation}
\Xi_{-}(Q \mid \mid P;f) \le E_Q(f)-E_P(f) \le \Xi_{+}(Q \mid \mid P;f).
\label{eq:In_dupuis}
\end{equation}where 
\begin{eqnarray}
\Xi_{+}(Q \mid \mid P;f)&=& \inf _{c>0}\left\{ \frac{1}{c}  \widetilde{\Lambda} _{P,f}(c) +  \frac{1}{c} R(Q \mid \mid P) \right\} 
\label{eq:goaldivplus} \\
\Xi_{-}(Q \mid \mid P;f)&=&\sup _{c>0}\left\{ -\frac{1}{c}  \widetilde{\Lambda} _{P,f}(-c) -  \frac{1}{c} R(Q \mid \mid P) \right\} .
\label{eq:goaldivminus}
\end{eqnarray}

We refer to \cite{dupuis2015path} and \cite{dupuis2011uq} for 
details of the proof but the main idea
behind the proof is the variational principle for the relative entropy: 
for bounded $f$ we have,
\cite{dupuisellis},
\[
\log E_P(e^f) = \sup_{Q} \left\{ E_Q(f) - R(Q \mid \mid P)\right\}  
\]
and thus for any $Q$
\[
E_Q(f) \le \log E_P(e^f) + R(Q\mid \mid P) \, .
\]
Replacing $f$ by $c (f - E_P(f))$ with $c>0$  and optimizing over $c$ yields the upper bound. The lower bound is derived in a similar manner. 

\bigskip
\noindent{\bf Robustness}:
These new information bounds were shown in \cite{dupuis2011uq} to be robust in the sense that the upper bound is attained 
when considering  all models $Q$ with  a specified uncertainty threshold given by  $R(Q\mid \mid P) \le \eta$. Furthermore, the parameter $c$ in the variational representations \eqref{eq:goaldivplus} and \eqref{eq:goaldivminus}
controls the degree of robustness with respect to the model uncertainty captured by $R(Q\mid \mid P) $. In a control or optimization context these bounds are also related to  H$^\infty$ control, \cite{dupuis2000}.

\bigskip

As the following result from \cite{dupuis2015path} shows, the quantities $\Xi_{+}$ and $\Xi_{-}$ are divergences similar to the  relative (R$\acute{e}$nyi ) entropy, the $\chi^2$ divergence and the Hellinger distance. Yet they depend on the observable $f$ and thus will be referred to as goal-oriented divergences.

\bigskip
\noindent
{\bf Properties of the goal-oriented divergence:}
\begin{enumerate}
\item $\Xi_{+}(Q \mid \mid P ;f ) \ge 0$ and $\Xi_{-}(Q \mid \mid P;f)\le 0$.
\item $\Xi_{\pm}(Q \mid \mid P ;f) = 0$ if and only if $Q=P$ or $f$ is constant P-a.s.
\end{enumerate}

It is instructive to understand the bound when $P$ and $Q$ are close 
to each other. Again we refer to \cite{dupuis2015path} for a proof and
provide here just an heuristic argument. First note that if $P=Q$ then 
it is easy to see that the infimum in the upper bound is attained at 
$c=0$ since $R(Q\mid\mid P)=0$ and $\widetilde{\Lambda} _{P,f}(c) > 0$ 
for $c>0$ (the function is convex in $c$ and we have $\widetilde{\Lambda} _{P,f}(0) = \widetilde{\Lambda} _{P,f}'(0)=0$ and $\widetilde{\Lambda} _{P,f}''(0)= \frac{1}{2} Var_P(f)$. 
So if $R(Q\mid\mid P)$ is small,  we can expand the right-hand side 
in $c$ and we need to find 
\[
\inf_{c > 0} \left\{ c \frac{Var_P(f)}{2} + O(c^2) + \frac{1}{c}R(Q\mid\mid P) \right\}\, .
\]
Indeed,  we find that the minimum has the form $\sqrt{Var_P(f)} \sqrt{2 R(Q\mid\mid P)} + O(R(Q \mid \mid P))$, \cite{dupuis2015path}.  The lower bound is similar
and we obtain:

\bigskip
\noindent
{\bf Linearized UQ bound} \cite{dupuis2015path}:
 If $R(P\mid\mid Q)$ is small we have 
\begin{eqnarray}
\mid E_Q(f)- E_P(f)\mid  \le   \sqrt{Var_P(f)}\sqrt{2R(Q\mid \mid P)}+O(R(Q\mid \mid P)).
\label{eq:goaldiv_linear}
\end{eqnarray}

\subsection{Example: Exponential Family} 
Next we  compute the goal-oriented divergences for an exponential family which covers many cases of interest 
including Markov and Gibbs measures (see Sections 3.4 and 4), as well as numerous probabilistic models in machine learning \cite{bishop,MacKay:2003}.

Given a reference measure $P^0$ (which does not 
need to be a finite measure) we say that $P^\theta$ is an exponential 
family if $P^\theta$ is absolutely continuous with respect to $P^0$ 
with 
\[
\frac{dP^\theta}{dP^0}(x) = \exp\left( t(x) \cdot \theta  -F(\theta)\right)
\]
where $\theta= [\theta_1, \cdots, \theta_K]^T  \in \Theta \subset \mathbb{R}^K$ 
is the parameter vector,  $t(x)=[t_1(x),..., t_K(x)]^T$ is the sufficient statistics 
vector and $F(\theta)$ is the log-normalizer 
\[
F(\theta) = \log \int e^{t(x) \cdot \theta}  dP^0(x)\,.
\]
Note that $F(\theta)$ is  a cumulant generating function for the 
sufficient statistics, for example we have $\nabla_\theta F(\theta) =  E_{P^\theta}(t)$.  
The relative entropy between two members of the exponential family
is then computed as   
\begin{eqnarray}
R(P^{\theta'} \mid \mid P^\theta) = \int \log \frac{dP^{\theta'}}{dP^\theta}(x)  dP^{\theta'}(x) 
&=&  E_{P^{\theta'}} \left( (\theta'-\theta)\cdot t(x) \right) +F(\theta)-F(\theta^{\prime}) \nonumber \nonumber \\ 
&=& (\theta'-\theta)\cdot \nabla F(\theta')+F(\theta)-F(\theta^{\prime})
\label{eq:re_exponential}
\end{eqnarray}
If we consider an observable which is a linear function of the sufficient statistics, that is  
\begin{equation}
f(x)= t(x) \cdot v 
\end{equation}
for some vector $v \in \mathbb{R}^K$ 
then the cumulant generating function of $f-E_{P^\theta}(f)$ is 
\begin{equation}
\tilde{\Lambda} _{P^\theta,f}(c) =\log E_{P^\theta}[e^{cf}]-c E_{P^\theta}(f)
= F(\theta+ cv)-F(\theta)-cv \cdot \nabla F(\theta).
\label{eq:cumu_exponential}
\end{equation}
and thus combining \eqref{eq:re_exponential} and \eqref{eq:cumu_exponential}, we obtain the divergences 
\begin{eqnarray}
\Xi_{+}(P^{\theta'}\|P^\theta; f)&=&\inf_{c>0}\frac{1}{c}\{ (\theta'-\theta)\cdot \triangledown F(\theta')-F(\theta^{\prime})+F(\theta+cv)-cv\cdot \triangledown F(\theta)\}
\\
\Xi_{-}(P^{\theta'}\|P^\theta;f)&=&\sup_{c>0}(-\frac{1}{c})\{ (\theta'-\theta)\cdot \triangledown F(\theta')-F(\theta^{\prime})+F(\theta-cv)+cv\cdot \triangledown F(\theta)\} 
\end{eqnarray}
%
Note that if our observable is not in the sufficient statistics class then we 
can  obtain a similar formula by simply enlarging the sufficient 
statistics to  include the observable in question.

\subsection{Example: IID sequences}\label{section:iid_2}

To illustrate the scaling properties of the goal-oriented divergence consider 
first two product measures $P_N$ and $Q_N$ as in Section \ref{section:iid_1} and the same sample mean observable \eqref{mean:sec2}.
We now apply the bounds \eqref{eq:goaldivplus} and \eqref{eq:goaldivminus}
to $N f_N = \sum_{k=1}^N g(\sigma_k)$ to obtain 
\[
\frac{1}{N} \Xi_{-}(Q_N \mid \mid P_N ;  N f_N) \le E_{Q_N}(f_N)- E_{P_N}(f_N) \le \frac{1}{N} \Xi_{+}(Q_N \mid \mid P_N; N f_N).
\]
The following lemma shows that the bounds scale correctly with $N$. 
\begin{lemma} We have 
\[
\Xi_{\pm}(Q_N \mid \mid P_N ;  N f_N) = N \Xi_{\pm}(Q \mid \mid P ; g) \,.
\]
\end{lemma}

\proof We have already noted that $R( Q_N \mid \mid P_N) = N R(Q \mid \mid P)$. 
Furthermore 
\begin{eqnarray}
{\widetilde \Lambda}_{P_N,Nf_N} (c)& = & \log E_{P_N}(e^{c NF_N }) - c E_{P_N}(N f_N) \nonumber \nonumber \\
&=&  \log \int_{\mathcal{X}_N} e^{c\sum_{i=1}^N g(\sigma_i)} \prod_{i=1}^N dP(\sigma_i) - c E_{P_N}\left( \sum_{i=1}^N g(\sigma_i)\right)\nonumber\nonumber \\
& = & N \log E_P (e^{cg}) - c N E_P(g) = N {\widetilde \Lambda}_{P,g} (c).
\end{eqnarray}
This result shows that the goal oriented divergence bounds captures perfectly 
the behavior of ergodic average as $N$ goes to infinity. In particular  
when P and Q are very close, $\Xi_{\pm}(Q \mid \mid P;g) \to 0$, which contrasts sharply with all the bounds in Section \ref{section:iid_1}.

\section{UQ and nonlinear response bounds for Markov sequences}
\label{section:markov_2}

In the context of Markov chains, there are a number of UQ challenges 
which are usually not addressed by standard numerical analysis or   UQ 
techniques: (a) Understand the effects of a model uncertainty 
on the long-time behavior (e.g. steady state) of the model. 
(b) Go beyond linear response and be able to understand how large perturbations affect the model, both in finite and long time regimes. 
(c) Have a flexible framework allowing to compare different models
as, for example for Ising model versus mean-field model approximations  considered in Section \ref{section:Examples}. 

The inequalities of Section \ref{section:goal-oriented} can provide new insights to all three questions,
at least when the bounds can be estimated or computed numerically or analytically.
As a first example in this direction we   consider Markov dynamics  with the same set-up as in Section 
\ref{section:markov_1}. We have the following bounds which 
exemplify how  the goal-oriented divergences provide  UQ bounds  for the long-time behavior of Markov 
chains.

\begin{theorem}
\label{thm:markov}
Consider two irreducible Markov chains with transition 
matrices $p(x,y)$ and $q(x,y)$ and stationary distributions $\mu_p$ and 
$\mu_q$ respectively.  If $p(x,\cdot)$ and $q(x,\cdot)$ are mutually absolutely continuous we have for any observable $g$ the bounds
\[
\xi_{-}(q \mid \mid p;g) \le E_{\mu_q}(g) - E_{\mu_p}(g) \le  \xi_{+}(q \mid \mid p;g)\, ,
\]
where 
\begin{eqnarray} \label{eq:bound for divergence rate}
\xi_{+}(q \mid \mid p;g) &=& \inf_{c \ge 0} \left\{\frac{1}{c}\lambda_{p,g}(c) + \frac{1}{c} r(q\mid \mid p) \right\}  \nonumber \nonumber \\
\xi_{-}(q \mid \mid p;g) &=& \sup_{c \ge 0} \left\{ -\frac{1}{c}\lambda_{p,g}(-c) - \frac{1}{c} r(q\mid \mid p)\right\}\, .
\end{eqnarray}
Here $$r(q\mid \mid p) = \lim_{N\to\infty}\frac{1}{N} R(Q_N\mid \mid P_N)$$
is the relative entropy rate and $\lambda_{p,g}(c)$ is the logarithm of the 
maximal eigenvalue of the non-negative matrix with entries $p(x,y)\exp( c ( g(y) - E_{\mu_p}(g)))$. 

Moreover, we have the asymptotic expansion in relative entropy rate $r(q\mid \mid p)$,
\begin{equation}
\label{linearization:theorem}
\left| E_{\mu_q}(g) - E_{\mu_p}(g) \right| \le \sqrt{ v_{\mu_p}(g) } \sqrt{2 r(q\mid \mid p)} + O (r(q \mid \mid p))
\end{equation}
where 
\[
v_{\mu_p(g)} \,=\, \sum_{k=-\infty}^\infty E_{\mu_p} 
\left( g(X_k) g(X_0)\right)
\]
is the integrated auto-correlation function for the observable $g$. 
\end{theorem}

\proof We apply the goal-oriented divergence bound to the observable $N f_N= \sum_{k=1}^N g(X_i)$ and have 
\[
\frac{1}{N} \Xi_{-}(Q_N \mid \mid P_N ;  N f_N) \le E_{Q_N}(f_N)- E_{P_N}(f_N) \le \frac{1}{N} \Xi_{+}(Q_N \mid \mid P_N; N f_N).
\]
We then take the limit $N \to \infty$. By the ergodicity of $P_N$ we have 
$\lim_{N\to \infty} E_{P_N}(f_N) =  E_{\mu_p}(g)$ and similarly for $Q_N$. 
We have already established in Lemma \ref{lem:scaling_markov} the existence 
of the limit $r(q\mid \mid p) \,=\, \lim_{N\to \infty} \frac{1}{N}R(Q_N \mid \mid P_N)$. 
For the moment generating function in $\Xi_{\pm}$ we have
\begin{eqnarray}
\frac{1}{N}\widetilde{\Lambda}_{P_N,Nf_N}(c) &=& \frac{1}{N} \log E_{P_N}(e^{c Nf_N})  - c \frac{1}{N} E_{P_N}(Nf_N) \nonumber \nonumber \\
&=& \frac{1}{N} \log \sum_{X_0, \cdots X_N}  \nu_p(X_0) \prod_{k=1}^N p(X_{k-1}, X_k) e^{c g(X_k)}  -  c E_{P_N}(f_N) \nonumber \nonumber \\ 
&=& \frac{1}{N} \log \nu_p P_{cg}^N  -  c E_{P_N}(f_N) \nonumber 
\end{eqnarray}
where $P_{cg}$ is the non-negative matrix with entries $p(x,y)e^{cg(y)}$. 
The Perron-Frobenius theorem gives the existence of the limit. 

The asymptotic expansion is proved exactly as for the linearized UQ bound \eqref{eq:goaldiv_linear}. 
It is not difficult to compute the second derivative of 
$\frac{1}{c}\lambda_{p,g}(c)$ with respect to $c$ by noting all function are 
analytic of function of $c$ and thus we can freely exchange the $N \to \infty$
limit with the derivative with respect to $c$. Therefore we obtain that 
\[
\frac{d^2}{dc^2} \lambda_{P,g}(0) \,=\, \lim_{N \to \infty} Var_{P_N}(Nf_N) 
\]
and a standard computation shows that the limit is the integrated 
autocorrelation function $v_{\mu_p}(g)$. 

\medskip\noindent
{\bf Remark:}
A well studied  case of UQ for stochastic models and in particular stochastic dynamics  is  linear response,
also referred to  as   local sensitivity analysis,  which addresses the role of infinitesimal perturbations to model parameters of  probabilistic models, e.g. \cite{Hairer:10,Glynn:book}. Here \eqref{linearization:theorem}
provides  computable bounds in the linear response regime, as demonstrated earlier in \cite{dupuis2015path} and which can be used for fast screening of uninfluential parameters in reaction networks with a very large number of parameters, \cite{arampatzis2015}.

\medskip
\noindent
{\bf Nonlinear response bounds:} Beyond  linear response considerations,  nonlinear response methods attempt to   address the role of larger parameter perturbations. Some of the relevant methods  involve   asymptotic series expansions  in terms of the  small parameter perturbation \cite{diezemann2012response,maes2015response},  which quickly become computationally intractable as more terms need to be computed. However, 
the inequalities \eqref{eq:goaldivplus} and \eqref{eq:goaldivminus}
provide robust and computable nonlinear response bounds.

The main result in Theorem \ref{thm:markov} was first obtained in  \cite{dupuis2015path}. Here we revisit it in the context of scalability in both space and time and  connect it 
to nonlinear response calculations for stochastic dynamics  in statistical mechanics. 
This connection is made more precise in the following Corollaries which follow from Theorem \ref{thm:markov} and provide  layers of progressively  simpler-and accordingly less sharp-bounds:

\begin{corollary} \label{corollary:RER}
Based on the assumptions  and definitions in Theorem \ref{thm:markov}, we have the following two bounds that involve two  upper bounds  of $r(q\mid \mid p)$.  Bound (i) is sharper than bound (ii), while the latter is straightforward to calculate analytically. 

(i) Let $R(q(x,\cdot)\|p(x,\cdot))=\sum_y q(x,y)\log  \frac{q(x,y)}{p(x,y)}$;  then, 
\begin{eqnarray} 
\xi_{+}(q \mid \mid p;g) &\le & \inf_{c \ge 0} \left\{\frac{1}{c}\lambda_{p,g}(c) + \frac{1}{c} \sup_x R(q(x,\cdot)\|p(x,\cdot)) \right\}  \nonumber  \\
\xi_{-}(q \mid \mid p;g) &\ge & \sup_{c \ge 0} \left\{ -\frac{1}{c}\lambda_{p,g}(-c) - \frac{1}{c} \sup_x R(q(x,\cdot)\|p(x,\cdot))\right\}\, .
\end{eqnarray}

(ii) Next,  we have  the upper bound in terms of the quantity $\sup\limits_{x,y}|\log  \frac{q(x,y)}{p(x,y)}|$,
\begin{eqnarray} 
\xi_{+}(q \mid \mid p;g) &\le & \inf_{c \ge 0} \left\{\frac{1}{c}\lambda_{p,g}(c) + \frac{1}{c} \sup_{x,y} |\log \frac{q(x,y)}{p(x,y)}| \right\}  \nonumber \\
\xi_{-}(q \mid \mid p;g) &\ge & \sup_{c \ge 0} \left\{ -\frac{1}{c}\lambda_{p,g}(-c) - \frac{1}{c} \sup_{x,y}|\log \frac{q(x,y)}{p(x,y)}| \right\}\, .
\end{eqnarray}

\end{corollary}

\begin{proof}
We consider the relative entropy rate $r(q\mid \mid p)$,
\begin{align}\label{eq:RER_bound1}
r(q\mid\mid p) &= \sum_{x,y} \mu_q(x) q(x,y)
\log  \frac{q(x,y)}{p(x,y)}  \nonumber \\
&=E_{\mu_q(x)}\left(\sum_y q(x,y)
\log  \frac{q(x,y)}{p(x,y)}\right) \nonumber\\
&=E_{\mu_q(x)}\left(R(q(x,\cdot)\|p(x,\cdot))\right)  \nonumber\\
& \le \sup_x R(q(x,\cdot)\|p(x,\cdot)),
\end{align}
where $R(q(x,\cdot)\|p(x,\cdot))=\sum_y q(x,y)
\log  \frac{q(x,y)}{p(x,y)}$.
Moreover, we have 
\begin{align*}
R(q(x,\cdot)\|p(x,\cdot))=\sum_y q(x,y)
\log  \frac{q(x,y)}{p(x,y)} \le \sup_y |\log  \frac{q(x,y)}{p(x,y)}|.
\end{align*}
Therefore  we can obtain another bound for $r(q\mid\mid p)$, that is,
\begin{equation}\label{eq:RER_bound2}
r(q\mid\mid p) \le \sup_{x,y} |\log  \frac{q(x,y)}{p(x,y)}|.
\end{equation}
This bound  may be not as sharp as the one in \eqref{eq:RER_bound1}, but it is more easily computable.
Thus, by \eqref{eq:bound for divergence rate}, \eqref{eq:RER_bound1} and \eqref{eq:RER_bound2}, it is easy to obtain (i) and (ii).
\end{proof}

If we consider the linearized bound in \eqref{linearization:theorem}, then combining the bounds \eqref{eq:RER_bound1} and \eqref{eq:RER_bound2} of $r(q\|p)$, we can obtain the following bound, which is a further simplification of Corollary \ref{corollary:RER},
 again at the expense of the tightness of the bounds.
\begin{corollary}
\label{corollary:RER_linear}
Under the assumptions  and definitions in Theorem \ref{thm:markov}, we have:
\begin{eqnarray}
\xi_{\pm}(q \mid \mid p;g) &\le& \pm \sqrt{ v_{\mu_p}(g) } \sqrt{2 \sup_x R(q(x,\cdot)\|p(x,\cdot))}+O(\sup_x R(q(x,\cdot)\|p(x,\cdot))) 
\label{eq:RER_linear1}\\
&\le& \pm \sqrt{ v_{\mu_p}(g) } \sqrt{2 \sup_{x,y} |\log  \frac{q(x,y)}{p(x,y)}|}+O(\sup_{x,y} |\log  \frac{q(x,y)}{p(x,y)}|)
\label{eq:RER_linear2}.
\end{eqnarray}
\end{corollary}

\medskip
\noindent{\bf Remark:}
By the previous two Corollaries, we get some cheap ways to replace the calculation of  $\xi_{\pm}(q \mid \mid p;g)$ since it is much easier to calculate $\sup\limits_x R(q(x,\cdot)\|p(x,\cdot))$ or $\sup\limits_{x,y} |\log  \frac{q(x,y)}{p(x,y)}|$  than $r(q\|p)$ itself, especially the latter one. 
In practice, we can first attempt to estimate $\xi_{\pm}(q \mid \mid p;g)$ by calculating the leading term in \eqref{eq:RER_linear1} or \eqref{eq:RER_linear2}.
If the  the linearization assumptions in the last Corollary fail, then we can try to use  Corollary \ref{corollary:RER} or Theorem \ref{thm:markov} which can also give computable bounds or estimates of $\xi_{\pm}(q \mid \mid p;g)$. 

Finally, 
the bound in \eqref{eq:RER_bound2} is the Markov chain  analogue of the triple norm $|||\cdot|||$ used  in the estimation of UQ bounds for QoIs  of Gibbs measures, which we discuss in depth in Section \ref{section:Gibbs}.

\section{UQ and nonlinear response  bounds for Gibbs measures}
\label{section:Gibbs}

The Gibbs measure is one of the central objects in statistical mechanics  and molecular dynamics simulation,  \cite{simon2014statistical}, \cite{tuckerman:book}.  On the other hand Gibbs measures in  the form of Boltzmann Machines or Markov Random Fields provide  one of the  key classes of models in machine learning and pattern recognition, \cite{MacKay:2003,bishop}. Gibbs measures are probabilistic models which are inherently high dimensional, describing spatially distributed systems or a large number of interacting molecules.
In this Section we derive scalable UQ bounds for Gibbs measures based on the goal oriented inequalities discussed in Section \ref{section:goal-oriented}. Gibbs measures can be  set on a lattice or in continuum space, here for  simplicity in the presentation we focus on lattice systems.

\smallskip
\noindent
{\bf Lattice spins systems.} We consider Gibbs measures for lattice 
systems on $\mathbb{Z}^d$. If we let $S$ be the configuration space of a 
single  particle at a single site $x \in \mathbb{Z}^d$, then $S^X$ is the 
configuration space for the particles in $X \subset \mathbb{Z}^d$; we 
denote by  $\sigma_X \,=\, \{ \sigma_x\}_{x\in X}$ an element of $S^X$.   
We will be interested in large systems so we let $\Lambda_N=\{x\in 
\mathbb{Z}^d, |x_i|\le n\}$ denote the square lattice with $N=(2n+1)^d$ 
lattice sites. We shall use the shorthand notation $\lim \limits_N$ to denote 
taking limit along the increasing sequence of lattices $\Lambda_N$ which 
eventually cover $\mathbb{Z}^d$. 

\smallskip
\noindent
{\bf Hamiltonians, interactions, and Gibbs measures}. 
To specify a Gibbs measures we specify the energy  
$H_N(\sigma_{\Lambda_N})$ of a set of particles
in the region $\Lambda_N$. It is convenient to introduce the concept of 
an interaction 
$\Phi= \{\Phi_X: X \subset \mathbb{Z}^d, X \textrm{finite}\}$ 
which associates to any finite subset $X$ a function $\Phi_X(\sigma_X)$ which 
depends  only on the configuration in $X$. We shall always assume that 
interactions are  translation-invariant, that is for any $X \subset 
\mathbb{Z}^d$ and any $a\in \mathbb{Z}^d$, $\Phi_{X+a}$ is obtained by 
translating $\Phi_X$.  For translation-invariant interactions we have the 
norm \cite{simon2014statistical}
\begin{equation}
\label{triple}
||| \Phi ||| = \sum_{\mathrm{X} \ni 0} |\mathrm{X}|^{-1}  \|\Phi_X\|_\infty
\end{equation}
and denote by $\mathcal{B}$  the corresponding Banach space of 
interactions. 
Given an interaction $\Phi$ we then define the Hamiltonians $H^\Phi_N$ (with free boundary conditions) by 
\begin{equation}
H_N^\Phi(\sigma_{\Lambda_N})=\sum_{X \subset \Lambda_N}\Phi_X(\sigma_X),
\end{equation}
and Gibbs measures $\mu_N^\Phi$ by 
\begin{equation}
d\mu_N^\Phi(\sigma_{\Lambda_N}) = \frac{1}{Z_N^\Phi} e^{ - H_N(\sigma_{\Lambda_N})} dP_N(\sigma_{\Lambda_N}) ,
\end{equation}
where $P_N$ is the counting measure on $S^{\Lambda_N}$ and $Z_N^\Phi= \sum_{ \sigma_{\Lambda_N}} e^{ - H_N(\sigma_{\Lambda_N})}$ is the normalization constant. 
In a similar way one can consider periodic boundary conditions 
or more general boundary conditions, 
see \cite{simon2014statistical} for details.

\medskip
\noindent
{\bf Example: Ising model.} For the $d$-dimensional nearest neighbor 
Ising model at inverse temperature $\beta$ we have 
\[
H_N(\sigma_{\Lambda_N})\,=\, -\beta J \sum_{\langle x, y \rangle \subset \Lambda_N} \sigma(x) \sigma(y) - \beta h \sum_{x \in \Lambda_N} \sigma(x)  
\]
where $\langle x, y \rangle $ denotes a pair of neighbors with $\sup\limits_{i}|x_i-y_i|=1$. So we have 
\[
\Phi_X = \left\{ \begin{array}{cc} -\beta J \sigma(x) \sigma(y)\, , & X = \{ x, y\}, \nonumber \\
- \beta h \sigma(x)\, , & X= \{x\},\nonumber \\
0 & \textrm{otherwise}, \nonumber \\
\end{array}
\right.
\]
and it is easy to see that \eqref{triple} becomes
\[
|||\Phi||| = \beta(d |J| + |h|).
\]

\smallskip 
\noindent{\bf Observables.}  
We will consider observables of the form \[
f_N(\sigma_{\Lambda_N})= \frac{1}
{N}\sum_{x \in \Lambda_N} g(\sigma_x)
\]
for some observable $g$. It will be useful to 
note that $N f_N$ is nothing but Hamiltonian $H_N^{\Gamma^g}$ for 
the interaction $\Gamma^g$ with  
\begin{equation}
\label{gamma:sec4}
\Gamma^g_{\{x\}}=g\, , \quad\mbox{and} \quad  \Gamma^g_X=0 \quad  \mbox{if} \quad X \not= \{x\}\, .   
\end{equation}

\smallskip
\noindent 
{\bf UQ bounds for Gibbs measures in finite volume.} Given two Gibbs 
measure $\mu_N^\Phi$ and  $\mu_N^\Psi$ straightforward computations show 
that for the relative entropy we have 
\begin{equation}\label{eq:re_gibbs}
R(\mu_N^\Psi \mid \mid \mu_N^\Phi) \,=\, \log Z_N^\Phi - \log Z_N^\Psi + E_{\mu_N^\Psi}( H_N^\Phi - H_N^\Psi) \,.
\end{equation}
while for the cumulant generating function we have
\begin{equation}\label{eq:cg_gibbs}
\widetilde{\Lambda}_{\mu_N^\Phi, Nf_N}(c)  
= \log Z_N^{\Phi - c \Gamma^g} - \log Z_N^\phi  - c E_{\mu_N^\Phi}(N f_N)
\end{equation}
and thus we obtain immediately from the results in Section \ref{section:goal-oriented}
\begin{proposition}(Finite volume UQ bounds for Gibbs measures) For two Gibbs measures $\mu_N^\Phi$ and $\mu_N^\Psi$ 
we have the bound
\begin{equation}\label{eq:finite_gibbs}
\frac{1}{N}
\Xi_{-}(\mu_N^{\Psi}\mid \mid \mu_N^{\Phi}; Nf_N) \le  E_{\mu_N^\Psi}(f_N) - E_{\mu_N^\Phi}(f_N) \le \frac{1}{N}
\Xi_{+}(\mu_N^{\Psi}\mid \mid \mu_N^{\Phi}; Nf_N)
\end{equation}
where 
\begin{align}
\Xi_{+}(\mu_N^{\Psi}\mid \mid \mu_N^{\Phi}; Nf_N) &= \inf_{c >0}\frac{1}{c}\left\{ \log Z_N^{\Phi - c \Gamma^g} - \log Z_N^\Psi +  E_{\mu_N^\Psi}( H_N^\Phi - H_N^\Psi) - c E_{\mu_N^\Phi}(N f_N)\right\} \label{eq:finite_gibbs+} \\
\Xi_{-}(\mu_N^{\Psi}\mid \mid \mu_N^{\Phi}; Nf_N)&= 
\sup_{c > 0}(-\frac{1}{c})\left\{ \log Z_N^{\Phi + c \Gamma^g} - \log Z_N^\Psi +  E_{\mu_N^\Psi}( H_N^\Phi - H_N^\Psi) +c E_{\mu_N^\Phi}(N f_N)\right\}
\label{eq:finite_gibbs-}.
\end{align}
\end{proposition}

\smallskip
\noindent 
{\bf UQ bounds for Gibbs measures in infinite volume.}
In order to understand how the bounds scale with $N$ we note first (see 
Theorem II.2.1 $\in$ \cite{simon2014statistical}) that the following limit 
exists 
\begin{equation}
p(\Phi) \,=\, \lim_{N} \log Z^\phi_N, 
\end{equation}
and $p(\Phi)$ is called the pressure for the interaction $\Phi$ (and is 
independent of the choice of boundary conditions).  
The scaling of the other terms in  the goal-oriented divergence $\Xi_{\pm}$ is slightly more delicate. In the absence of first order 
transition for the Gibbs  measure for the interaction $\Psi$ the finite volume Gibbs 
measures $\mu_N^\Psi$ have a well-defined unique limit $\mu^\Phi$ 
on $S^{\mathbb{Z}^d}$ which is translation invariant and ergodic with 
respect to $\mathbb{Z}^d$ translations. In addition we have (see Section III.3 in \cite{simon2014statistical})
\[
\lim_N \frac{1}{N} E_{\mu_N^\Psi} (H_N^\Phi) = E_{\mu^\Phi}( A^\Phi) \, \quad \textrm{ with  }  A^\Phi = \sum_{X\ni 0}\frac{1}{|X|}\Phi_X
\]
and moreover $E_{\mu^\Phi}( A^\Phi)$ can also be interpreted in terms of 
the derivative of the pressure functional
\[
E_{\mu^\Phi}( A^\Phi) = -\frac{d}{d\alpha} p( \Psi+\alpha \Phi) \mid_{\alpha=0}.
\]
We obtain therefore the following theorem which is valid in the 
absence of first order phase transitions. 

\begin{theorem} (Infinite-volume UQ bounds for Gibbs measures.) Assume that both $\Phi$ and $\Psi$ have a unique infinite-
volume Gibbs measure $\mu^\Phi$ and $\mu^\Psi$. Then we have the bound
\[
\xi_{-}(\mu_\Phi \mid \mid \mu_\Psi ; g)\le E_{\mu^\Phi}(g) - E_{\mu^\Psi}(g) \le \xi_{+}(\mu_\Phi\mid \mid \mu_\Psi ; g)
\]
where $\Gamma^g$ is given by \eqref{gamma:sec4} and, 
\begin{eqnarray}
\xi_{+}(\mu_\Phi \mid \mid \mu_\Psi ; g)&=& \inf_{c>0 }\frac{1}{c} \left\{ p(\Phi-c\Gamma^g) - p(\Psi) - \frac{d}{d\alpha}
p(\Psi+\alpha(\Phi-\Psi))\mid_{\alpha=0}  - c \frac{d}{dc} p(\Phi- c\Gamma^g)\mid_{c=0}  \right\} \nonumber  \\
\xi_{-}(\mu_\Phi \mid \mid \mu_\Psi ; g)&=&\sup_{c>0 }\frac{1}{c} \left\{ -p(\Phi+c\Gamma^g) + p(\Psi) + \frac{d}{d\alpha}
p(\Psi+\alpha(\Phi-\Psi))\mid_{\alpha=0}  + c\frac{d}{dc} p(\Phi+ c \Gamma^g)\mid_{c=0}  \right\} 
\nonumber
\end{eqnarray}
\end{theorem}

\smallskip
\noindent
{\bf Phase transitions.}  The bound is useful  even in 
the presence of first order phase transition which manifests itself
by the existence of  several infinite volume Gibbs measure 
consistent with the finite volume Gibbs measure (via the DLR condition)
or equivalently by the lack of differentiability of the pressure 
functional $p(\Phi+ \alpha \Upsilon)$ for some interaction $\Upsilon$. 
For example in the 2-d Ising model discussed in Section \ref{section:Examples}, below the critical 
temperature the pressure $p(\Phi)$ is not differentiable in $h$ at 
$h=0$: there are two ergodic infinite volume Gibbs measures which 
corresponds to the two values of left and right derivatives of the 
pressure (aka the magnetization). If necessary, in practice one 
will select a particular value of the magnetization, see the examples 
in Section \ref{section:Examples}.

\medskip\noindent{\bf  UQ bounds  and  the use of the triple norm $|||\Phi|||$.
}
It is not difficult to show (see Proposition II.1.1C and Lemma II.2.2C in \cite{simon2014statistical} and the definition of the triple norm in 
\eqref{triple}, that 
\begin{equation}
| \log Z_N^\Phi - \log Z_N^\Psi| \le \|H^\Phi_N - H_N^\Psi\|_\infty \le N |||\Phi-\Psi |||\,.
\end{equation}
and thus by \eqref{eq:re_gibbs} we have 
\begin{equation}
\label{eq:bound:triple}
\frac{1}{N} R(\mu_N^\Phi \mid \mid \mu_N^\Psi) \le 2 |||\Phi-\Psi|||\, .
\end{equation}
Therefore,  we obtain the bounds
\begin{eqnarray}
\Xi_+ &\le & \inf_{c>0} \left\{ \frac{1}{c}\widetilde{\Lambda}_{\mu_N^\Phi, Nf_N}(c) + \frac{2}{c} 
|||\Psi - \Phi|||\right\} \nonumber \nonumber \\
\Xi_- & \ge & \sup_{c>0} \left\{-\frac{1}{c}\widetilde{\Lambda}_{\mu_N^\Phi, Nf_N}(c) - \frac{2}{c} |||\Psi - \Phi|||\right\}\, .\nonumber 
\end{eqnarray}
These new upper and lower bounds, although they are less sharp, they still scale correctly in system size, while they are intuitive in capturing the dependence of the 
model discrepancy on  the fundamental level of the interaction discrepancy $|||\Psi - \Phi|||$; finally the bounds  do not require 
the computation 
of the relative entropy, due to upper bound  \eqref{eq:bound:triple}.

\medskip\noindent
{\bf Remark:} On the other hand, it is tempting but nevertheless misguided to try to bound $\widetilde{\Lambda}_{\mu_N^\Phi, Nf_N}(c)$ in terms of interaction norms. Indeed we have the bound   
$\frac{1}{c}\widetilde{\Lambda}_{\mu_N^\Phi, Nf_N}(c) \le \| N f_N - E_{\mu_N^\Phi}(Nf_N)\|_\infty$. 
But this bound becomes trivial:  since the the infimum over $c$ is then 
attained at $c=\infty$ with the trivial result that $\Xi_+(\mu_N^\Phi \mid \mid
\mu_N^\Psi; Nf_N) \le \| N f_N - E_{\mu_N^\Phi}(Nf_N)\|_\infty$ which is independent of $\Psi$ and thus useless.  

\medskip
\noindent
{\bf Linearized bounds.}  Applying the linearized bound \eqref{eq:goaldiv_linear}
to the Gibbs case gives the bound
\begin{equation}
\frac{1}{N}\Xi_{\pm}(\mu^\Psi_N\|\mu^\Phi_N; Nf_N)= \pm 
\sqrt{\frac{1}{N}Var_{\mu^\Phi_{N}}\big( \sum_{x\in \Lambda_N} g(\sigma_x)\big) }
\sqrt{\frac{2}{N}R(\mu^\Psi_{N}\|\mu^\Phi_{N})}+O(\frac{1}{N}R(\mu^\Psi_{N}\|\mu^\Phi_{N})).
\label{eq:linear_approx0}
\end{equation}
In the large $N$ limit, in the absence of first order transition, and if the spatial 
correlations in the infinite volume Gibbs state decays sufficiently fast then 
the variance term converges to the integrated auto-correlation function
\begin{eqnarray}
\lim_{N}\frac{1}{N}Var_{\mu^\Phi_{N}}\bigg( \sum_{x\in \Lambda_N} g(\sigma_x)\bigg)  &=& \sum_{x\in \mathbb{Z}^d} E_{\mu^\Phi} \left( (g(\sigma_x) - E_{\mu^\Phi}(g) )(g(\sigma_0) - E_{\mu^\Phi}(g) )\right) \nonumber \nonumber \\
&=& \frac{d^2}{dc^2}P( \Phi - c \Gamma^g) \mid_{c=0}  
\end{eqnarray}
which is also known as susceptibility in the statistical mechanics literature. 

Finally, we get a  simple and easy to implement linearized UQ bound when we replace \eqref{eq:bound:triple} in \eqref{eq:linear_approx0},
namely
\begin{equation}
\frac{1}{N}\Xi_{\pm}(\mu^\Psi_N\|\mu^\Phi_N; Nf_N)= \pm 2
\sqrt{\frac{1}{N}Var_{\mu^\Phi_{N}}\big( \sum_{x\in \Lambda_N} g(\sigma_x)\big) }
\sqrt{|||\Psi - \Phi|||}+O(|||\Psi - \Phi|||).
\label{eq:linear_approx3}
\end{equation}
Each one of terms on the right hand side of \eqref{eq:linear_approx3}
can be either computed using Monte Carlo simulation or can be easily estimated, see for instance the calculation of $|||\Psi - \Phi|||$ in the Ising case earlier.

\section{UQ for  Phase Diagrams of molecular systems }
\label{section:Examples}

In this section, we will consider the Gibbs measures for one and two-dimensional  Ising  and mean field models, which are exactly solvable models, see e.g. 
\cite{baxter2007exactly}. We also note that mean field models can be obtained as a minimizer of relative entropy  in the sense of \eqref{variational:intro},  where $P$ is an Ising model and ${\cal Q}$ is a parametrized family of  product distributions, \cite{MacKay:2003}.

Here  we will demonstrate the  use of the goal-oriented divergence, discussed earlier in Section \ref{section:goal-oriented} and Section \ref{section:Gibbs},
to  analyze  uncertainty quantification for sample mean observables such as the mean magnetization 
\begin{equation}
\label{mag:sec5}
f_N=\frac{1}{N}\sum_{x \in \Lambda_N} \sigma(x),
\end{equation} in different phase diagrams based on these models.  We use exactly solvable models as a test bed for the accuracy of  our bounds, and demonstrate their tightness even in phase transition regimes.  In  \ref{section:backgroup_Ising and mean}, we give some background about one/two-dimensional 
Ising models and mean field models and recall some well-known formulas.

\smallskip
\noindent
{\bf Implementation of the UQ bounds} 
The results in Sections \ref{section:markov_2} and \ref{section:Gibbs} demonstrate mathematically  that the bounds relying on the goal oriented divergences $\Xi_{\pm}$ are the only available ones that scale properly for long times and high dimensional systems. Therefore we turn our attention to the implementation of these bounds. First we note that   the bounds depending of the triple norms $|||\cdot |||$, as well as the the linearized bounds of Section \ref{section:Gibbs} provide implementable upper bounds, see also the strategies in  \cite{arampatzis2015} for the linearized regime, which are related to sensitivity screening.

By contrast, here we focus primarily on exact calculations of the goal oriented divergences $\Xi_{\pm}$, at least for cases where either the Ising 
models are exactly solvable or in the case where the known (surrogate) model is a  mean field approximation.
We denote by $\mu_N$ the Gibbs measures of the model we assume to be known 
and  $\mu_N'$ the Gibbs measure of the  model we try to estimate.   Then from 
\eqref{eq:finite_gibbs}--\eqref{eq:finite_gibbs-}, recalling that 
$\Lambda_{\mu_N,Nf_N}( c)= \tilde{\Lambda}_{\mu_N,Nf_N}(c) + c E_{\mu_{N}}(Nf_N)$,  we can rewrite the  bounds  as

\begin{eqnarray}
&& \hspace{.5cm} E_{\mu'_{N}}(f_N)\ge 
\sup_{c>0}\left\{-\frac{1}{cN}\Lambda_{\mu_N,Nf_N}(-c)-\frac{1}{cN}R(\mu'_{N}\mid\mid\mu_{N})\right\}  
\nonumber \nonumber \\
&& \hspace{.5cm} E_{\mu'_{N}}(f_N) 
\le \inf_{c>0}\left\{\frac{1}{cN}\Lambda_{\mu_N,Nf}(c)+\frac{1}{cN}R(\mu'_{N} \mid \mid \mu_{N})\right\} \nonumber
\end{eqnarray}
and obtain an explicit formula for each term in the large $N$ limit in 
terms of the pressure, mean energy and magnetization for the models.  
In the figures below we will display the upper and lower bounds 
using simple optimization algorithm in Matlab to find the 
optimal $c$ in the bounds. Note that in the absence of exact formulas 
we would need to rely on numerical  sampling of those quantities, an issue we will discuss elsewhere. 

For completeness and for comparison with the exact bounds 
we will also use and display the approximate linearized bounds
\begin{eqnarray}
E_{\mu'_N}(f_N) &\gtrapprox&  E_{\mu_{N}}(f_N) - \sqrt{\frac{1}{N}Var_{\mu_{N}}(Nf_N)}\sqrt{\frac{2}{N}R(\mu'_{N} \mid \mid \mu_{N})} 
 \nonumber  \\ \hspace{3cm} E_{\mu'_{N}}(f_N) &\lessapprox& E_{\mu_{N}}(f_N) + \sqrt{\frac{1}{N}Var_{\mu_{N}}(Nf_N)}\sqrt{\frac{2}{N}R(\mu'_{N} \mid \mid \mu_{N})} \nonumber 
\end{eqnarray}
where each term is computable in the large $N$ limit in terms of the pressure, susceptibility, magnetization, and so on.

\subsection{Three examples of UQ  bounds for Phase Diagrams}
Next we consider three cases where our methods provide  exact UQ bounds for phase diagrams between two high dimensional probabilistic models.  Here we compare three classes of Gibbs measures for Ising models. (1) Mean field models with different parameters well beyond the linear response regime, (2) Ising models compared to their mean field approximations, and (3) Ising models with vastly different parameters. 
All these examples cannot be handled with conventional arguments such as linear response theory because they fall into two categories: either,  (a)  the models have parameters differing significantly, for instance by  at least  $50\%$, or  (b) the comparison is between  different models, e.g. a  complex model and  a simplified surrogate model which is  a potentially  inaccurate   approximation such as the mean field   of the original Ising model.

\smallskip
\noindent
{\bf (1) Mean field versus mean field models.\,} 
Firstly, we consider  two mean field models, assume $\mu_{N;mf}$ and $\mu'_{N;mf}$ are their Gibbs measures (probabilities) defined in  \ref{section:backgroup_Ising and mean} with $h_{mf}=h+dJm$ and $h'_{mf}=h'+dJ'm'$, respectively.
By some straightforward calculation in  \ref{section: example_goal oriented divergence calculation }, we obtain the ingredients of the UQ bounds discussed earlier in the Section:
\begin{eqnarray}
\frac{1}{N}R(\mu'_{N;mf}\|\mu_{N;mf})= \log \frac{e^{\beta h_{mf}} +e^{-\beta h_{mf}}}{e^{-\beta' h'_{mf}}+e^{\beta' h'_{mf}}}+(\beta' h'_{mf}-\beta h_{mf})m',
\end{eqnarray}
\begin{eqnarray}
\frac{1}{N}\Lambda_{\mu_{N;mf},Nf_N}(c)  =\log \frac{e^{(c+\beta h_{mf})} +e^{-(c+\beta h_{mf})}}{e^{-\beta h_{mf}}+e^{\beta h_{mf}}},
\label{eq:cumulant_mf}
\end{eqnarray}
and
\begin{eqnarray}
\frac{1}{N}Var_{\mu_{N;mf}}(\sum_{x \in \Lambda_N} \sigma(x)) 
& = & 1- m^2,
\label{eq:var_mf}
\end{eqnarray}
where $m$ and $m'$  are  the magnetizations \eqref{mag:sec5}
of these two
mean field models and  can be obtained by solving the implicit equation \eqref{eq:mean_mf}.
Here we note that  the solution of the equation \eqref{eq:mean_mf} when $h=0$  has a super-critical pitchfork bifurcation. In our discussion regarding   mean field vs mean field and  1-d Ising  vs mean field models  we only consider the upper branch of the stable solution. But, in our discussion about  2d Ising vs mean field, we consider the both upper and lower branches.

In  \ref{section:backgroup_Ising and mean},  for given parameters, we can calculate the magnetizations, the goal-oriented divergence bounds  and their corresponding linearized bounds which we use in deriving exact formulas for the UQ bounds.
Indeed, for Figure \ref{Fig:1dmfVmf_beta1}, we set $J=2$ and 
consider the Gibbs measure of the 1-d  mean field model with $h=0$ as the benchmark and plot the magnetization based on this distribution as a function of inverse temperature $\beta$.
Then, we perturb the external magnetic field to $h=0.6$ and consider the Gibbs measure with this external magnetic field.
We  plot the goal-oriented divergence  bounds  of the magnetization of the Gibbs measure with $h=0.6$ as a function of $\beta$ as well as their corresponding linearized approximation in this figure. To test the sharpness of these bounds, we also give the magnetization with $h=0.6$ in the figure. We can see that the bounds work well here. The upper bound almost coincides with the magnetization.  The linearized approximation  works well at low temperature, but, it does not work as well as the goal-oriented bound around the critical point. The reason for this is that relative entropy between those two measures is bigger here due to  the bigger perturbation of $h$ and linearization is a poor approximation of the bounds. Also, by the figure, for $h=0$, we can 
 see that the magnetization of vanishes at high temperatures. At low temperatures it goes to its maximum value $m=1$. For non-zero $h$, we see that there is no phase transition and the magnetization increases gradually from close to $m=0$ at  high temperatures ($\beta \ll 1$)  to $m=1$ at low temperatures ($\beta \gg 1$). 

In Figure \ref{Fig:1dmfVmf_h2}, we set $J=1$ and consider the Gibbs measure of the 1-d mean field model with $\beta=1$ as the benchmark and plot the magnetization based on this measure as a function of $h$ in the figure. Then we perturb $\beta$ by $60\%$ and obtain another Gibbs measure with $\beta=1.6$ that has a phase transition at $h=0$. 
In the figure, we give  the upper/lower goal-oriented divergence bounds of the magnetization based on the Gibbs measure with  $\beta=1.6$ as well as their corresponding linearized bounds. 
To test the sharpness of the bounds, we also plot the magnetization with $\beta=1.6$ as a function of $h$.
The goal-oriented divergence bounds work well here.  We can see the upper bound almost coincide with the magnetization when $h$ is positive and the lower bound almost coincide with the magnetization when $h$ is negative.  Similarly with \ref{Fig:1dmfVmf_beta1}, the linearized bounds  make a relatively poor estimation around the critical point $h=0$ because of the bigger relative entropy between these two measures.

\begin{figure}[H]
	\centering
	\begin{subfigure}[b]{0.45\textwidth}
		\centering
		\includegraphics[width=\linewidth]{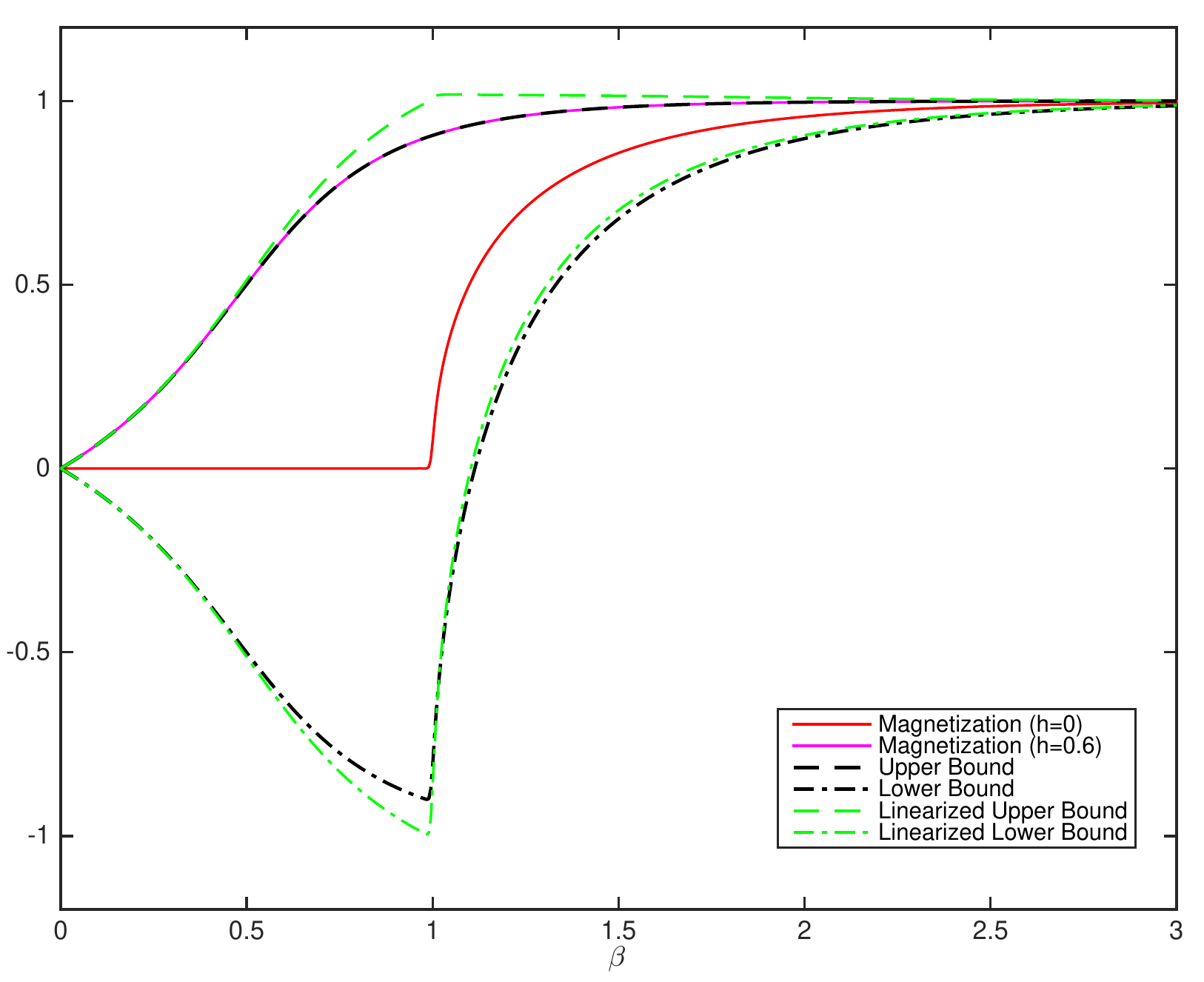}
		\caption{}
		\label{Fig:1dmfVmf_beta1}
	\end{subfigure}%
	\begin{subfigure}[b]{0.45\textwidth}
		\centering
		\includegraphics[width=\linewidth]{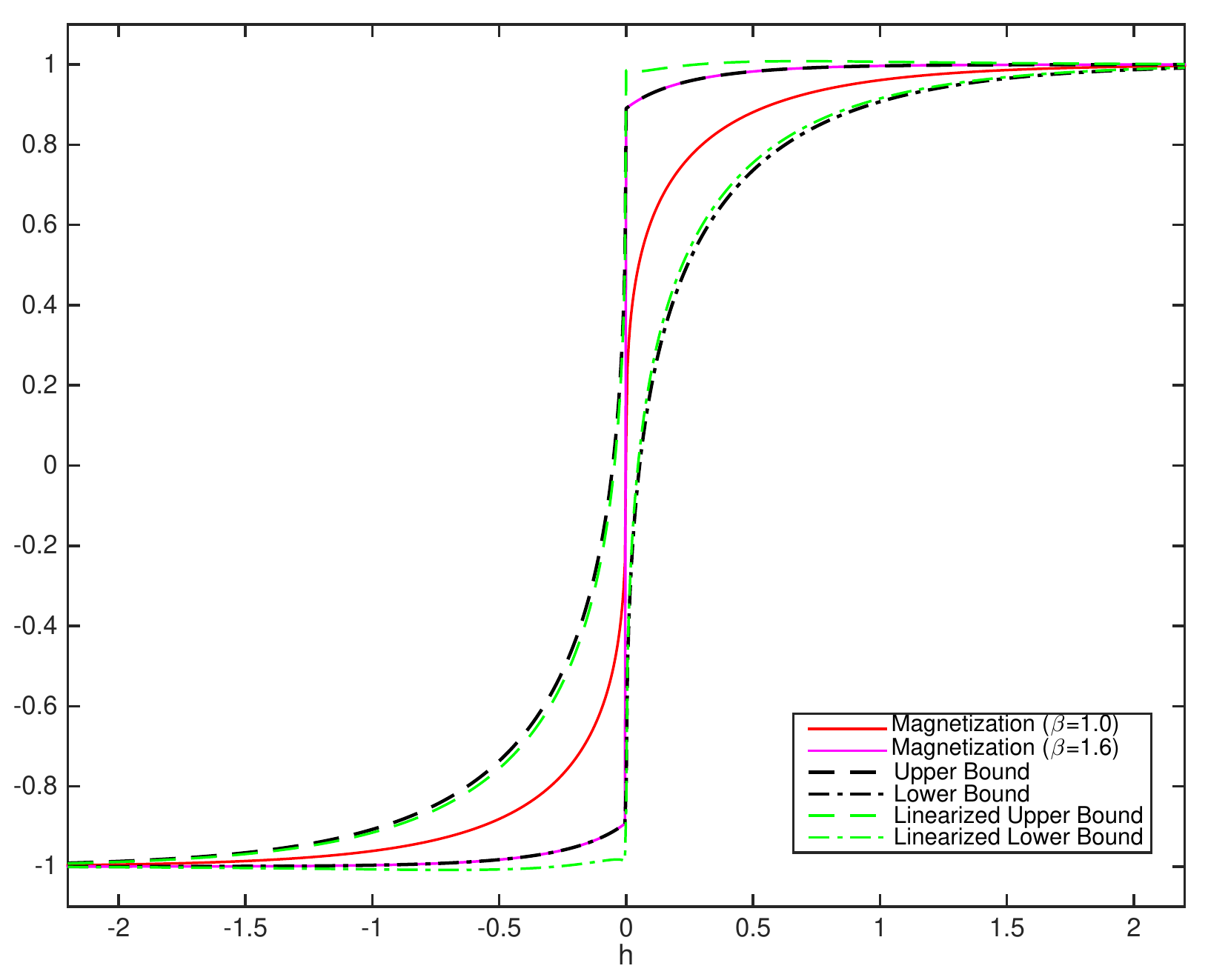}
		\caption{}
		\label{Fig:1dmfVmf_h2}
	\end{subfigure}
	\caption{\subref{Fig:1dmfVmf_beta1}: The red solid line is the magnetization for $h=0$; The magenta solid line is the magnetization for $h=0.6$. 
The black dashed/dash-dot lines is the upper/lower UQ bounds given by the goal-oriented divergences of the magnetization \eqref{mag:sec5}
for $h=0.6$ . The green dashed/dash-dot line is the linearized upper/lower bound. \subref{Fig:1dmfVmf_h2}: The red solid line is the magnetization for $\beta=1$; The magenta solid line is the magnetization for $\beta=1.6$. 
The black dashed/dash-dot lines is the upper/lower goal-oriented divergence bound of the magnetization for $\beta=1.6$ . The green dashed/dash-dot line is the linearized upper/lower bound.}
	\label{Fig:1dmfVmf}
\end{figure}


\smallskip
\noindent
{\bf (2a) One-dimensional  Ising model versus mean field.\,} 
Consider the $1$-d Ising model and mean field model and assume $\mu_N$ and $\mu_{N;mf}$ are respectively their Gibbs distributions,
which are defined in  \ref{section:backgroup_Ising and mean} . 
Then, by straightforward calculations, we obtain 
\begin{align}
\lim \limits_{N} \frac{1}{N}R(\mu_N\|\mu_{N;mf}) 
&=\log \frac{e^{\beta [h+Jm]}+e^{-\beta [h+Jm]}}{e^{\beta J}\cosh(\beta h)+k_1}+\frac{\beta J}{k_1}(k_1-\frac{2e^{-2\beta J}}{e^{\beta J}\cosh(\beta h)+k_1}-me^{J\beta}\sinh(h\beta))
\end{align}
where $k_1=\sqrt{e^{2J\beta}\sinh^2(h\beta)+e^{-2J\beta}}$;
detailed calculations can be found in  \ref{section: example_goal oriented divergence calculation } .
By \eqref{eq:cumulant_mf} and \eqref{eq:var_mf}, we have
\begin{equation}
\frac{1}{N}\Lambda_{{N;mf},Nf_N}(c) 
  =  \log \frac{e^{[c+\beta (h+Jm)] } +e^{-[c+\beta (h+Jm)]}}{e^{-\beta [h+Jm]}+e^{\beta [h+Jm]}}.
\end{equation}
and
\begin{eqnarray}
\frac{1}{N}Var_{\mu_{N;mf}}(\sum_{x \in \Lambda_N} \sigma(x))  
& = & 1- m^2.
\end{eqnarray}
Combining with  \ref{section:backgroup_Ising and mean},  for given parameters, we can calculate the magnetizations, the goal-oriented divergence bounds and their corresponding linearized approximation. 

In Figure \ref{Fig:Ising_1d_beta}, we set $h=0$ and $J=1$ and consider the  Gibbs measure of the mean field model as the benchmark, that is we use it as a surrogate model for  the Ising model. 
In the figure, we see that its magnetization vanishes at high temperatures. At low temperatures it goes to its maximum value $m=1$, exhibiting spontaneous magnetization and a phase transition at the inverse temperature $\beta=1$.
We  plot the upper/lower goal-oriented divergence  bound as well as their corresponding linearized bounds of the magnetization as a function of  $\beta$. 
To test the sharpness of these bounds, we also plot the magnetization of the Ising model in the figure. 
The magnetization of the Ising model vanishes for all temperatures, exhibiting no phase transitions. In this sense the mean field approximation of the Ising model is a very poor one and the UQ bounds depicted in Figure \ref{Fig:Ising_1d_beta}
capture and quantify the nature of this approximation. Indeed, we can see that the bounds work well here, but the linearized lower bound fails for low temperatures because of the considerable difference between $\mu_N$ and $\mu_{N;mf}$.
In Figure \ref{Fig:Ising_1d_h}, we set $\beta=1$ and $J=1$ and consider the bounds and the magnetizations as  a function of the external field $h$.
Similarly with  Figure \ref{Fig:Ising_1d_beta}, we take the Gibbs measure of the mean field model as the benchmark. To test the sharpness of the bounds, we also plot the magnetization of the Ising model in the figure. We can see the goal-oriented divergence bounds of the magnetization of the Ising model works well here.  The upper bound almost coincides with  it for positive $h$ and the lower bound almost coincide with it for negative $h$. However, the linearized ones do not give a good approximation around the point $h=0$.

    \begin{figure}[H]
	\centering
	\begin{subfigure}[b]{0.45\textwidth}
		\centering
		\includegraphics[width=\linewidth]{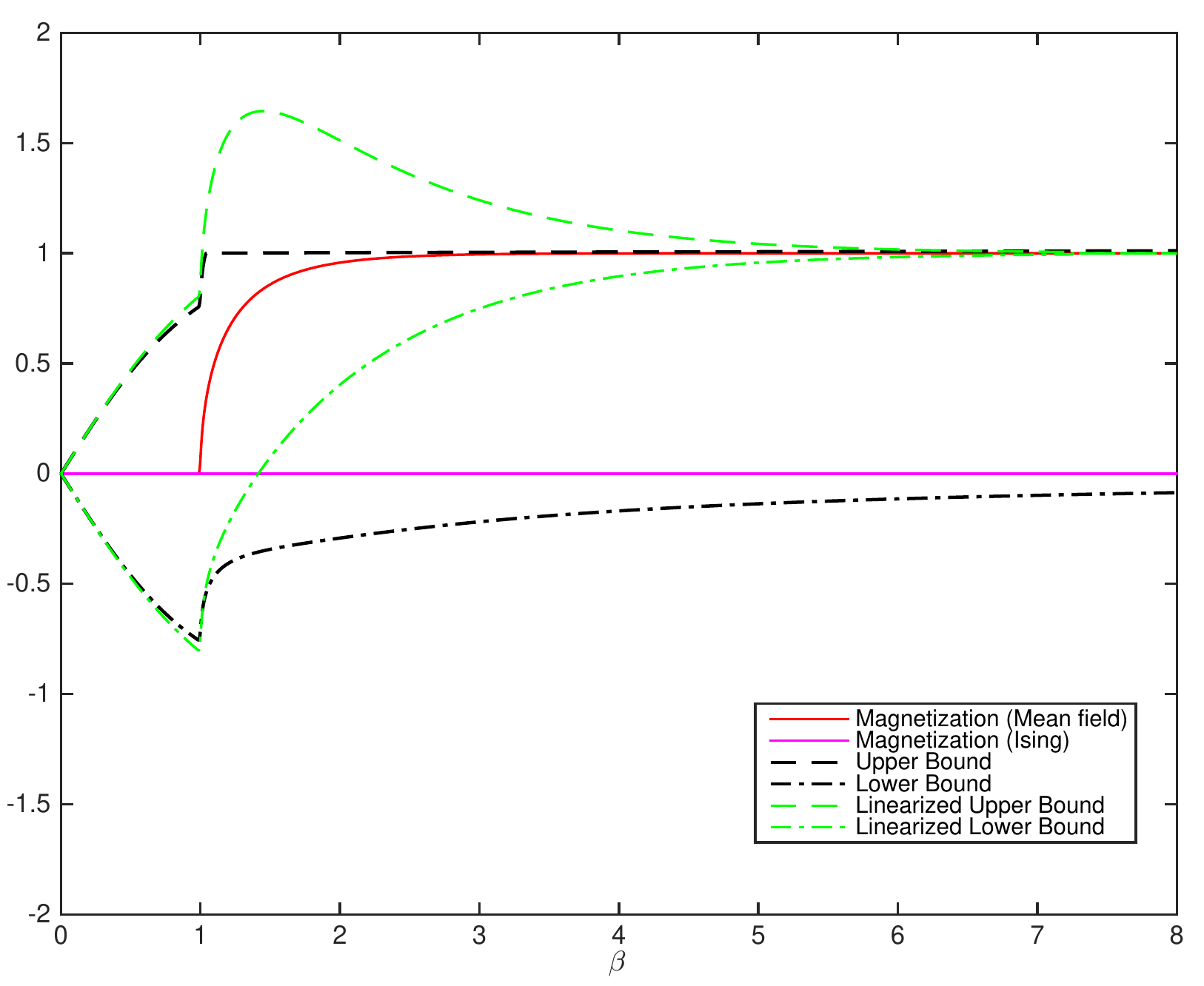}
		\caption{}
		\label{Fig:Ising_1d_beta}
	\end{subfigure}%
	\begin{subfigure}[b]{0.45\textwidth}
		\centering
		\includegraphics[width=\linewidth]{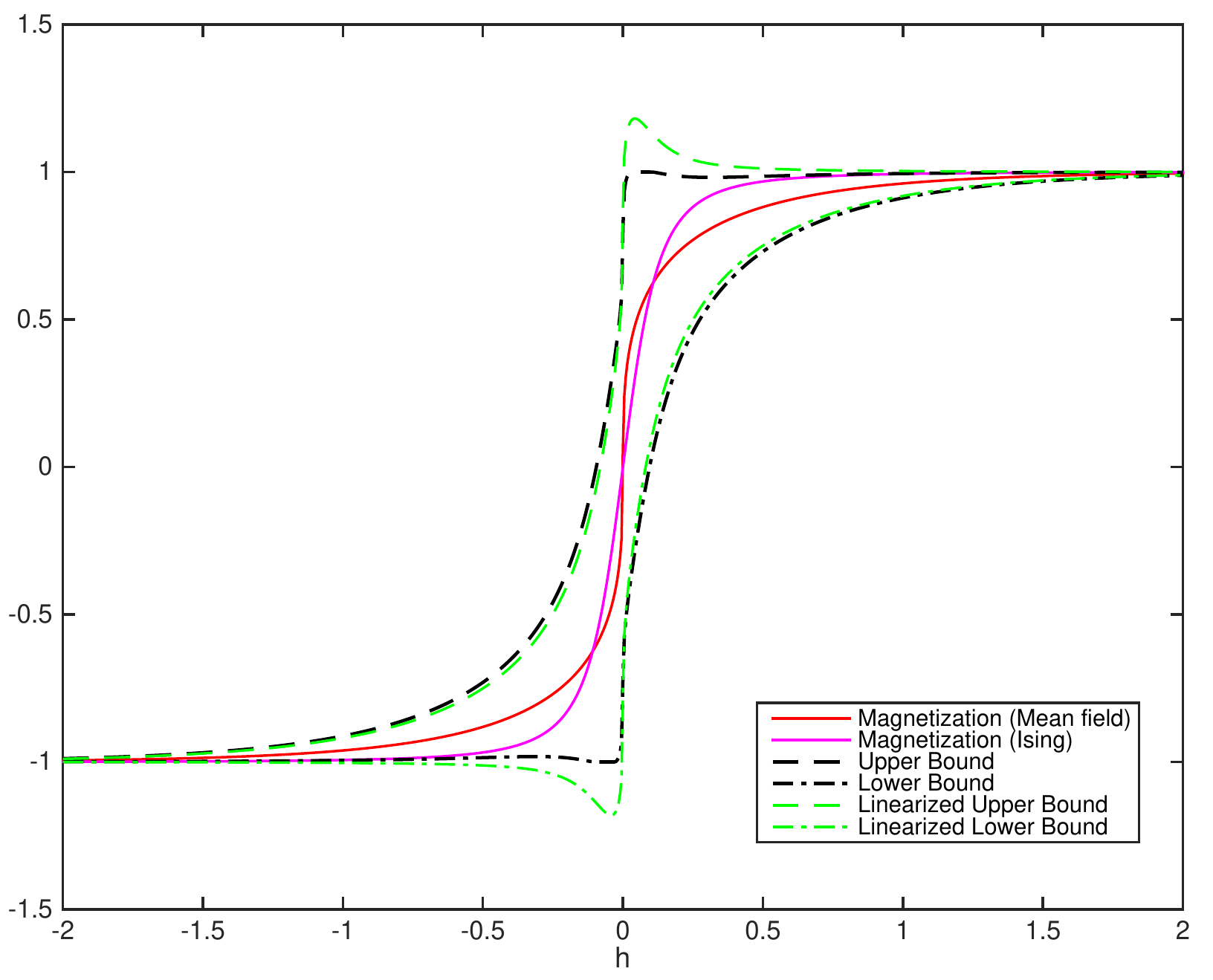}
		\caption{}
		\label{Fig:Ising_1d_h}
	\end{subfigure}
	\caption{\subref{Fig:Ising_1d_beta}:The red solid line is the magnetization of 1-d mean field model for $h=0$; The magenta solid line is the magnetization of 1-d Ising model for $h=0$. 
The black dashed/dash-dot lines is the upper/lower goal-oriented divergence bound of the magnetization of the Ising model. The green dashed/dash-dot line is the linearized  upper/lower bound. \subref{Fig:Ising_1d_h}:The red solid line is the magnetization of 1-d mean field model for $\beta=1.0$; The magenta solid line is the magnetization of 1-d Ising model $\beta=1.0$. 
The black dashed/dash-dot lines is the upper/lower goal-oriented divergence bound of the magnetization of the Ising model. The green dashed/dash-dot line is the linearized  upper/lower bound. }
	\label{Fig:Ising_1d}
\end{figure}

\smallskip
\noindent
{\bf (2b) Two-dimensional  Ising model 
versus mean field.\,} We revisit the example in (2a) above but this time in two dimensions where the Ising model exhibits phase transitions at a finite temperature.
WE denote by  $\mu_N$ and $\mu_{N;mf}$  the  Gibbs distributions  for the two-dimensional  zero-field Ising model and two-dimensional   mean field model with $h_{mf}=2Jm$, respectively.
Then, by  straightforward calculations, we obtain 
\begin{eqnarray}
\lim \limits_{N}\frac{1}{N}R(\mu_N\|\mu_{N;mf})
&=&\log [e^{-2\beta Jm}+e^{2\beta Jm}]-\frac{\log2}{2}-\frac{1}{2\pi}\int_0^{\pi}\log[\cosh^2(2\beta J)+k(\theta)]d\theta \nonumber \nonumber \\
&+& \beta J \frac{\sinh(4\beta J)}{\pi}\int_0^{\pi} \frac{1}{k(\theta)}[1-\frac{1+\cos(2\theta)}{\cosh^2(2\beta J)+k(\theta)}]d\theta-2\beta JmM_0,
\end{eqnarray}
\begin{equation}
\frac{1}{N}\Lambda_{\mu_{N;mf},Nf_N}(c) 
  =  \log \frac{e^{(c+2\beta Jm) } +e^{-(c+2\beta Jm)}}{e^{-2\beta Jm}+e^{2\beta Jm}}
\end{equation}
and 
\begin{eqnarray}
\frac{1}{N}Var_{\mu_{N;mf}}(\sum_{x \in \Lambda_N} \sigma(x))  
& = & 1- m^2,
\end{eqnarray}
where $m$ and $M_0$ are the spontaneous magnetizations of the two-dimensional  mean field model and Ising models, respectively and can be obtain by solving \eqref{eq:mean_mf} and \eqref{eq:mean_Ising2}.
Detailed calculations can be found in  \ref{section: example_goal oriented divergence calculation }.
Combining with \ref{section:backgroup_Ising and mean},  for given parameters, we can calculate the magnetizations, the goal-oriented divergence bounds and their corresponding linearized approximation. 

In Figure \ref{Fig:Ising_2d_beta}, we set $h=0$ and $J=1$ and plot the bounds and the magnetizations  as a function of inverse temperature $\beta$.
Similarly with Figure \ref{Fig:Ising_1d}, we take the Gibbs measure of the mean field as the benchmark and consider the bounds for the magnetization of the Ising model.
We can see that the goal-oriented bounds work well here, especially in low temperatures. Notice the large uncertainty prior to the onset of the spontaneous magnetization (phase transition) which is due  to a pitchfork bifurcation and the two branches (upper and lower) reported in   Figure 1b, as well as in the  panels in Figure 4.
The linearized bounds  also work well, but they are not as sharp as the goal-oriented divergence bounds around the critical points because of the larger value of the relative entropy $R(\mu_N\|\mu_{N;mf})$. There are  phase transitions for both mean field model and Ising model. The critical points are $1/2$ and $\log(1+\sqrt{2})/2$ for mean field model and Ising model, respectively. Both their magnetizations vanish at  high temperatures  and  go to their maximum values $1$ at low temperature. 

Actually, the spontaneous magnetizations we consider in Figure \ref{Fig:Ising_2d_beta} are both based on the definition $M=\lim \limits_{h\to 0^{+}}\langle \sigma(x) \rangle$. If we consider the definition $M=\lim \limits_{h\to 0^{-}}\langle \sigma(x) \rangle$, we can obtain another figure which is Figure \ref{Fig:Ising_2d_beta1}.
We can see the quantities in Figure \ref{Fig:Ising_2d_beta1} are just the opposite of the corresponding quantities in Figure \ref{Fig:Ising_2d_beta}. Combining both figures gives us the uncertainty region reported in the Introduction.

  \begin{figure}[H]
	\centering
	\begin{subfigure}[b]{0.45\textwidth}
		\centering
		\includegraphics[width=\linewidth]{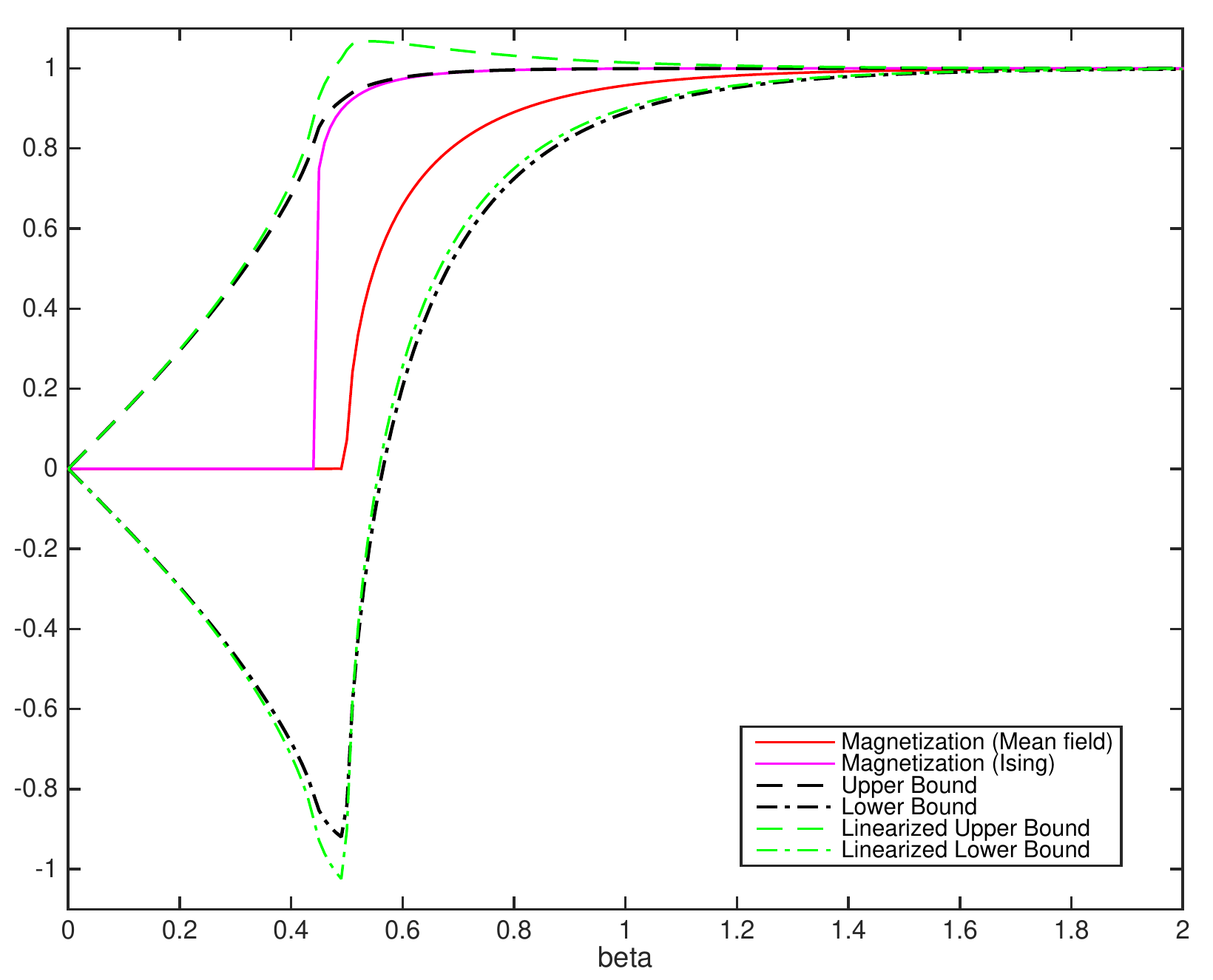}
		\caption{}
		\label{Fig:Ising_2d_beta}
	\end{subfigure}%
	\begin{subfigure}[b]{0.45\textwidth}
		\centering
		\includegraphics[width=\linewidth]{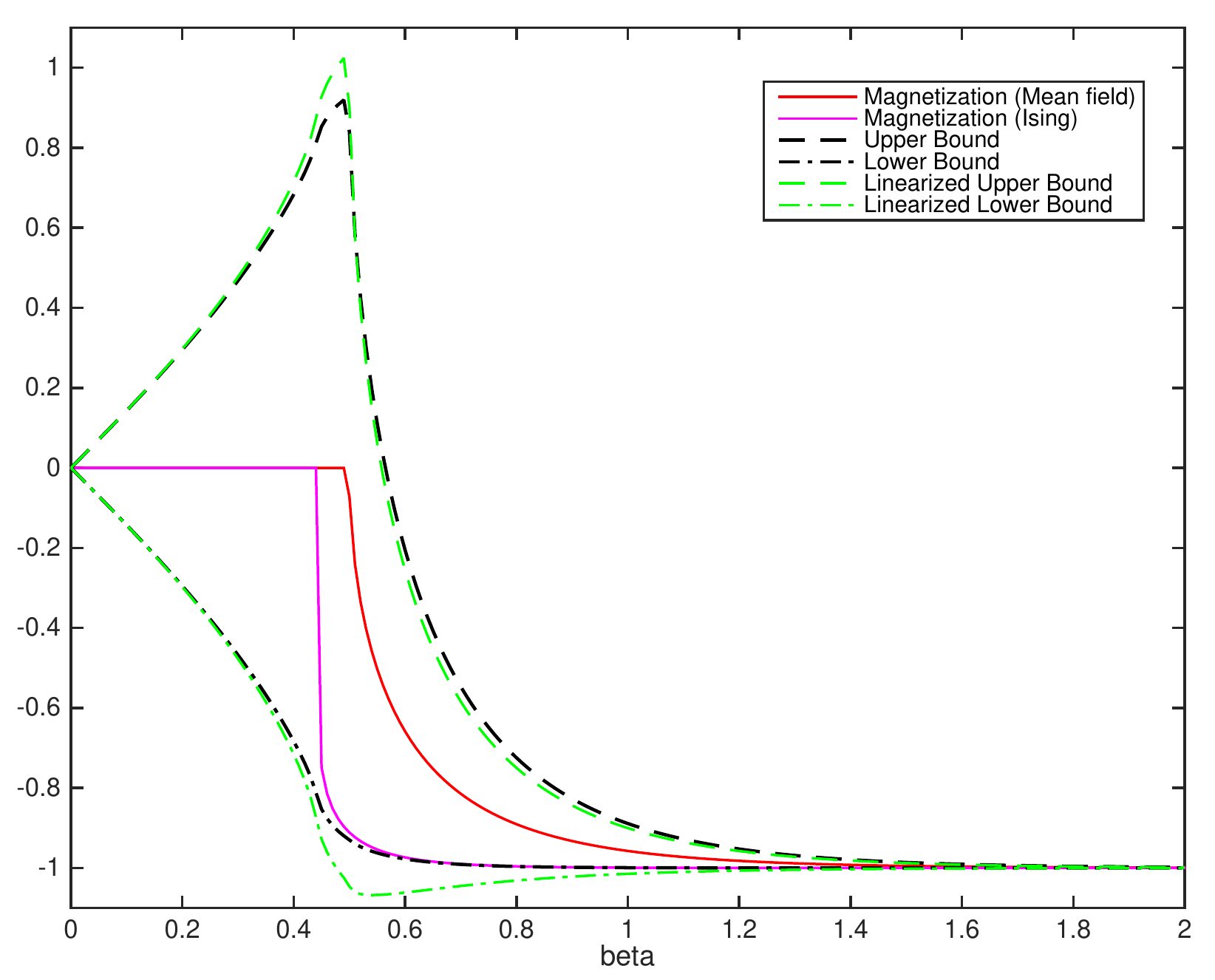}
		\caption{}
		\label{Fig:Ising_2d_beta1}
	\end{subfigure}
	\caption{ \subref{Fig:Ising_2d_beta} 
    The red solid line is the spontaneous magnetization of the 2-d mean field model with $h=0^{+}$; The magenta solid line is the spontaneous magnetization of 2-d Ising model with $h=0^{+}$;
The black dashed/dash-dot lines is the upper/lower goal-oriented divergence bound of the magnetization for Ising model; The green dashed/dash-dot line is the linearized  upper/lower bound.  \subref{Fig:Ising_2d_beta1}
The red solid line is the spontaneous magnetization of 2-d mean field model with $h=0^{-}$; The magenta solid line is the spontaneous magnetization of 2-d Ising model with $h=0^{-}$;
The black dashed/dash-dot lines is the upper/lower goal-oriented divergence bound of the magnetization for Ising model; The green dashed/dash-dot line is the linearized  upper/lower bound. }
	\label{Fig:Ising_2d}
\end{figure}
    
\smallskip
\noindent
{\bf (3) One-dimensional  Ising model versus Ising model.\,} 
Consider  two one-dimensional Ising models and $\mu_N$ and $\mu_N'$ are their Gibbs distributions  defined in \ref{section:backgroup_Ising and mean}.
By  straightforward calculation, we have 
\begin{align}
\lim \limits_{N}\frac{1}{N}R(\mu'_N\|\mu_N)
&=\log\frac{e^{\beta J}\cosh(\beta h)+\sqrt{e^{2J\beta}\sinh^2(h\beta)+e^{-2J\beta}}}{e^{\beta' J'}\cosh(\beta' h')+\sqrt{e^{2J'\beta'}\sinh^2(h'\beta')+e^{-2J'\beta'}}} \nonumber\nonumber \\
&+(\beta' J'-\beta J)(1-\frac{1}{k'_1}\frac{2e^{-2\beta' J'}}{e^{\beta' J'}\cosh(\beta' h')+k'_1})+(\beta' h'-\beta h)
\frac{1}{k'_1}e^{J'\beta'}\sinh(h'\beta')
\end{align}
and
\begin{equation}
\lim \limits_{N}\frac{1}{N}Var_{\mu_N}(\sum_{x \in \Lambda_N}\sigma(x))=\frac{1}{k_1^3}e^{-J\beta}\cosh(h\beta),
\end{equation}
where $k'_1=\sqrt{e^{2J'\beta'}\sinh^2(h'\beta')+e^{-2J'\beta'}}$.
The cumulant generating function is 
\begin{eqnarray}
\lim \limits_{N} \frac{1}{N}\Lambda_{\mu_N}(c) 
& = &\log \frac{e^{\beta J}\cosh(\beta h+c)+\sqrt{e^{2J\beta}\sinh^2(h\beta+c)+e^{-2J\beta}}}{e^{\beta J}\cosh(\beta h)+\sqrt{e^{2J\beta}\sinh^2(h\beta)+e^{-2J\beta}}},
\end{eqnarray}
detailed calculations can be found in  \ref{section: example_goal oriented divergence calculation }.
Combining with \ref{section:backgroup_Ising and mean},  for given parameters, we can calculate the magnetizations, the bounds given by goal-oriented divergence and their corresponding linearized approximation. 

In Figure \ref{Fig:1dIVI_beta}, we set  $J=1$ and plot the magnetizations of 1-d Ising model as a function of inverse temperature $\beta$ for $h=0$ and $h=0.6$, respectively. For the zero-field Ising model, used here as our benchmark, the magnetization vanishes for all temperatures. For $h=0.6$, the magnetization increases gradually to its maximum $1$. Clearly the models are far apart but the UQ bounds work remarkably well.
Indeed, we   plot the upper/lower goal-oriented divergence bound of the magnetization  for the nonzero-field Ising model.  The upper bound almost coincides with the magnetization itself. The lower bound is poor due to the symmetry of the bounds in $h$. If we break the symmetry by comparing models for different positive  external fields both bounds become much sharper (not shown).
The linearized bounds give a good approximation at high temperatures. However, at low temperatures, they are  not as sharp as the goal-oriented divergence bounds. This is due to the  larger  relative entropy $R(\mu\|\mu')$ between $\mu$ and $\mu'$.
In Figure \ref{Fig:1dIVI_h2},  we plot the magnetization of  the one-dimensional  Ising model as a function of $h$ for two different inverse temperatures $\beta=1$ and $\beta=1.6$. The  parameter $J$ was set to $1$. We also plot the upper/lower goal-oriented divergence bounds for  $\beta=1.6$. Similarly with Figure~\ref{Fig:1dIVI_beta}, we also plot the linearized  upper/lower bound in the figure.
The goal-oriented divergence bounds work well here.  We can see the upper bound almost coincides with the magnetization when $h$ is positive and the lower bound almost coincides with the magnetization when $h$ is negative.  
The linearized bounds  make a relatively poor estimation around the  since models are far apart due to the large perturbation in the parameter $\beta$ or $h$ in Figure 5.

\begin{figure}[H]
	\centering
	\begin{subfigure}[b]{0.45\textwidth}
		\centering
		\includegraphics[width=\textwidth]{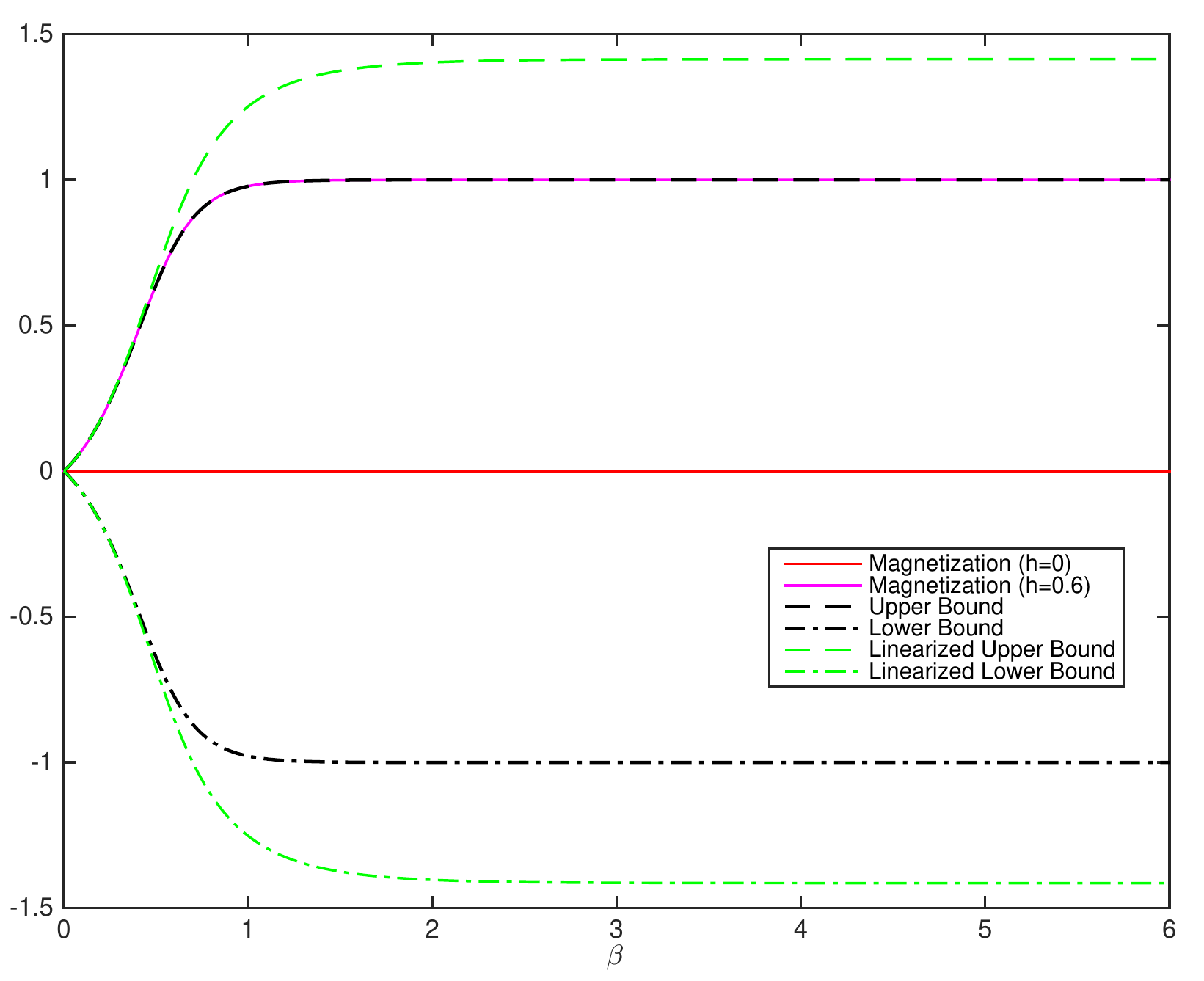}
		\caption{}
		\label{Fig:1dIVI_beta}
	\end{subfigure}%
	\begin{subfigure}[b]{0.45\textwidth}
		\centering
		\includegraphics[width=\textwidth]{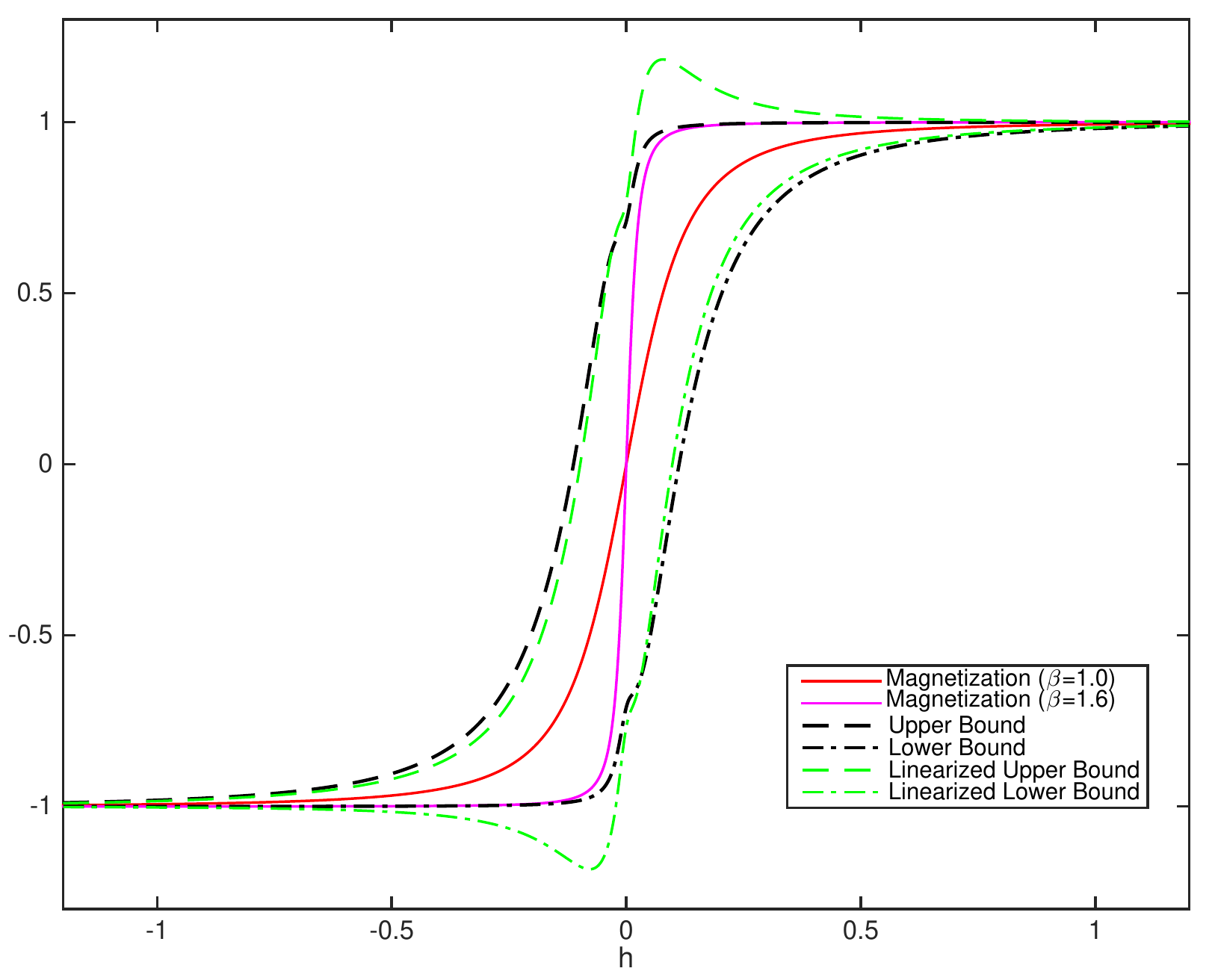}
		\caption{}
		\label{Fig:1dIVI_h2}
	\end{subfigure}
	\caption{\subref{Fig:1dIVI_beta} The red solid line is the magnetization of 1-d Ising model for $h=0$; the magenta solid line is the magnetization of 1-d Ising model for $h=0.6$;
the black dashed/dash-dot lines is the upper/lower bound by goal-oriented divergence; the green dashed/dash-dot line is the linearized  upper/lower bound. \subref{Fig:1dIVI_h2} The red solid line is the magnetization of 1-d Ising model for $\beta=1$; the magenta solid line is the magnetization of 1-d Ising model for $\beta=1.6$;
the black dashed/dash-dot lines is the upper/lower bound by goal-oriented divergence; the green dashed/dash-dot line is the linearized  upper/lower bound. }
	\label{Fig:1dIVI}
\end{figure}



\section{Conclusions}

In this paper we first showed  that the classical information 
inequalities such as Pinsker-typer inequalities and other 
inequalities based on the Hellinger distance, the $\chi^2$-divergence, or the R$\acute{e}$nyi  
divergence perform poorly for the purpose of    controlling QoIs of systems with many degrees of freedom, 
and/or in long time regimes. On the other hand we demonstrated that the goal oriented 
divergence introduced in \cite{dupuis2015path} scales properly and 
allows to control QoIs  provided they can be written
as ergodic averages or spatial averages,  e.g. quantities like autocorrelation, mean magnetization, specific  
energy, and so on. We illustrated the potential  of our approach by 
computing uncertainty quantification bounds for phase diagrams for Gibbs measures, that is for systems in the thermodynamic limit.
We showed that the bounds perform remarkably well even in the presence of 
phase transitions. 

Although we provided computable bounds and exact calculations,
there is still a lot to be done towards developing efficient Monte Carlo samplers
for the goal oriented divergences $\Xi_{\pm}$, which is a central mathematical object in our approach.
An additional  strength of our approach is that it also 
applies  to non-equilibrium systems which do not necessarily satisfy 
detailed balance, providing robust nonlinear response bounds. The key insight here is to study the statistical 
properties of the paths of the systems and to use the thermodynamic 
formalism for space-time Gibbs measures.  
Our results can be applied   to a wide range of problems in  statistical 
inference, coarse graining of complex systems, steady states 
sensitivity analysis for non-equilibrium systems, and  Markov 
random fields.   

\section*{Acknowledgments} 
The research of MAK and JW was partially supported by the Defense Advanced Research Projects
Agency (DARPA) EQUiPS program under the grant  W911NF1520122. The research of LRB was partially supported by the National Science Foundation (NSF) under the grant DMS-1515712.

\section*{Reference}
\bibliographystyle{unsrt}      
\bibliography{ref} 

\newpage

\appendix
\section{ Hellinger-based Inequalities}
\label{section:Hellinger}
\begin{lemma}
Suppose $P$ and $Q$ be two probability measures on some measure space 
$(\mathcal{X},\mathcal{A})$ and let $f : \mathcal{X} \to \mathbb{R}$ be
some quantity of interest (QoI), which is  measurable and has second moments with respect to both $P$ and $Q$. 
 Then
\begin{eqnarray}
 \left| E_{Q}(f)-E_{P}(f)\right|
& \le & \sqrt{2}H(Q, P) \sqrt{ Var_P(f) + Var_Q(g) + \frac{1}{2}(E_Q(f) - E_P(f))^2}. 
\end{eqnarray}
\end{lemma}
\begin{proof}
By Lemma $7.14$ in \cite{dashti2013bayesian},we have 
\begin{eqnarray*}
\left| E_{Q}(f)-E_{P}(f) \right|
& \le & \sqrt{2}H(Q, P) \sqrt{E_P (f^2) +E_Q (f^2)}.
\end{eqnarray*}
For any $c\in \mathbb{R}$, replace $f$ by $f-c$, we have 
\begin{align*}
\mid E_{Q}(f)-E_{P}(f) \mid &=\mid E_{P}(f-c)-E_{Q}(f-c)\mid \\
&\le \sqrt{2}H(Q, P) \sqrt{E_P ((f-c)^2) +E_Q ((f-c)^2)}.
\end{align*}
Thereby,
\begin{equation*}
\mid E_{Q}(f)-E_{P}(f) \mid \le \inf \limits_{c}\sqrt{2}H(Q, P) \sqrt{E_P ((f-c)^2) +E_Q ((f-c)^2)}
\end{equation*}
By some straight  calculations, we can find the optimal c is : 
\begin{equation*}
c^{\star}= \frac{E_P(f)+E_Q(f)}{2}.
\end{equation*}
Thus, we have 
\begin{eqnarray*}
\mid E_{Q}(f)-E_{P}(f) \mid 
& \le & \sqrt{2}H(Q, P)\sqrt{E_Q [(f-c^\star)^2] +E_P [(f-c^\star)^2]} \\ 
& = &  \sqrt{2}H(Q, P) \sqrt{ Var_P(f) + Var_Q(g) + \frac{1}{2}(E_Q(f) - E_P(f))^2}.
\end{eqnarray*}
\end{proof}
\section{Proof of Lemmas \ref{lem:scaling_iid} and \ref{lem:scaling_markov}}
\label{app:iid_markov}

\subsection{I.I.D. sequences}
{\bf Proof of Lemma \ref{lem:scaling_iid}.} since $P_N$ and $Q_N$ are 
product measures  we have $\frac{dP_N}{dQ_N}(\sigma_{\Lambda_N}) = \prod_{j=1}^N \frac{dP}{dQ}(\sigma_j)$. 

For the relative entropy we have 
\begin{eqnarray}
R(P_N\mid \mid Q_N) & = &\int_{\mathcal{X}_n} \log\frac{dP_N}{dQ_N}(\sigma_{\Lambda_N}) dP_N(\sigma_{\Lambda_N}) = \int_{ \mathcal{X}_n} \sum_{j=1}^N \log \frac{dP}{dQ}(\sigma_j) dP_N(\sigma_{\Lambda_N}) \nonumber \nonumber \\
&=& \sum_{j=1}^N \int_{\mathcal{X}}\log \frac{dP}{dQ}(\sigma_j) dP(\sigma_j)=  N R(P\mid\mid Q) \,.
\end{eqnarray}

For the Reny relative entropy we have 
\begin{eqnarray}
D_\alpha(P_N\mid \mid Q_N) & = & \log \int_{\mathcal{X}_n} \left(\frac{dP_N}{dQ_N}(\sigma_{\Lambda_N})\right)^{\alpha} dQ_N(\sigma_{\Lambda_N}) = \log \int_{ \mathcal{X}_n} \prod_{j=1}^N \left(\frac{dP}{dQ}(\sigma_j)\right)^{\alpha} dQ_N(\sigma_{\Lambda_N}) \nonumber \\
&=& \sum_{j=1}^N\log \int_{\mathcal{X}} \left(\frac{dP}{dQ}(\sigma_j)\right)^{\alpha} dQ(\sigma_j)=  N D_\alpha(P\mid\mid Q) \,.
\end{eqnarray}

For the $\chi^2$ distance we note first that 
\[
\chi^2(P \mid \mid Q) = \int \left(\frac{dP}{dQ}-1\right)^2 dQ =  
\int \left( \left(\frac{dP}{dQ}\right)^2- 2 \frac{dP}{dQ} + 1\right) dQ = 
\int  \left(\frac{dP}{dQ}\right)^2 dQ  - 1 \,,
\]
and therefore we have 
\begin{eqnarray}
\chi^2(P_N \mid \mid Q_N) & =& \int_{\mathcal{X}_n} \left( \prod_{j=1}^N \frac{dP}{dQ}(\sigma_j) \right)^2 dQ_N(\sigma_{\Lambda_N}) -1  = \prod_{j=1}^N \int_{\mathcal{X}} \left( \frac{dP}{dQ}(\sigma_j) \right)^2 dQ(\sigma_j)  -1 \nonumber \nonumber \\
&=& \left( 1 + \chi^2(P\mid\mid Q)\right)^N -1 \,.
\end{eqnarray}

For the Hellinger distance we note first that 
\[
H^2(P,Q) 
=  \int \left( \sqrt{\frac{dP}{dQ}} -1\right)^2 dQ 
= \int \left( \frac{dP}{dQ}- 2 \sqrt{\frac{dP}{dQ}} + 1 \right) dQ 
= 2 - 2\int \sqrt{\frac{dP}{dQ}} dQ \,,
\]
and thus $\int \sqrt{\frac{dP}{dQ}} dQ = 1 -\frac{1}{2} H^2(P,Q)$. Therefore we have
\begin{eqnarray}
H^2(P_N,Q_N)&=& 2 - 2\int_{\mathcal{X}_n}   \sqrt{\prod_{j=1}^N \frac{dP}{dQ}(\sigma_j)} dQ(\sigma_{\Lambda_N})  = 2 - 2\prod_{j=1}^N \int_{\mathcal{X}} \sqrt{\frac{dP}{dQ}(\sigma_j)}  dQ(\sigma_j) \nonumber \nonumber \\
&=&  2 - 2 \left( 1 - \frac{H^2(P,Q)}{2}\right)^N \,.
\end{eqnarray}
This concludes the proof of Lemma \ref{lem:scaling_iid}. \qed

\subsection{Markov sequences}

{\bf Proof of Lemma \ref{lem:scaling_markov}}:  The convergence of the relative 
entropy rate is well known and we give here a  short  proof for the convenience of 
the reader. 

Recall that $\nu_p$ and $\nu_q$ are the initial distributions of the Markov 
chain at time $0$ with transition matrices $p$ and $q$ respectively. 
We write $\nu_p^k$ the distribution at time $k$ as a row vector 
and we have then $\nu_p^k(x) \equiv \nu_p p^k (x)$ where $p^k$ is the 
matrix product.

By expanding the logarithm and integrating we find  
\begin{eqnarray}
&& \frac{1}{N} \int \log \frac{dP_N}{dQ_N} dP_N \nonumber \nonumber \\
&& = \frac{1}{N} \sum_{x_0, \cdots x_N} \log
\left(
\frac{\nu_p(x_0) p(x_0,x_1) \cdots p(x_{n-1}, x_n)}{\nu_q(x_0) q(x_0,x_1) \cdots q(x_{n-1}, x_n)}
 \right) 
 \nu_p(x_0) p(x_0,x_1) \cdots p(x_{n-1}, x_n) 
\nonumber  \nonumber \\
&& = \frac{1}{N} \sum_{x_0}  \log \frac{\nu_p(x_0)}{\nu_q(x_0)} \nu_p(x_0) +  \frac{1}{N}\sum_{k=1}^N 
\sum_{x_0, \cdots, x_k} 
\nu_p(x_0) p(x_0,x_1) \cdots p(x_{k-1}, x_k)   \log \frac{p(x_{k-1}, x_k)}{q(x_{k-1}, x_k)} \nonumber 
\nonumber \nonumber \\
&& = \frac{1}{N} \sum_{x_0} \log \frac{\nu_p(x_0)}{\nu_q(x_0)} \nu_p(x_0) + \frac{1}{N}\sum_{k=1}^N  \sum_{x,y}  \nu_p^k(x)  p(x,y)  \log \frac{p(x, y)}{q(x, y)}. \label{eq:rerate}
\end{eqnarray}
The first term goes to $0$ as $N \to \infty$ while for the second term, by the ergodic theorem we have 
that for any initial condition $\nu_p$,   $\lim_{N \to \infty} \frac{1}{N}\sum_{k=1}^{N} \nu_p^k= \mu_p$ where $\mu_p$ is stationary distribution.  Therefore we obtain that   
\[
\lim_{N \to \infty} \frac{1}{N} \int \log \frac{dP_N}{dQ_N} dP_N = \sum_{x,y}  \mu_p (x)  p(x,y)  \log \frac{p(x, y)}{q(x, y)} .
\]
Finally we note that the limit can be written as a relative entropy itself, since 
\[
\sum_{x,y}  \mu_p (x)  p(x,y)  \log \frac{p(x, y)}{q(x, y)} = \sum_{x} \mu_p(x) R\left(p(x, \cdot) \mid \mid  q(x, \cdot)\right) .
\]
As a consequence the relative entropy rate vanishes if and only if  
$ R\left(p(x, \cdot) \mid \mid  q(x, \cdot)\right) = 0$ for every $x$
that is if and only if $p(x,y)=q(x,y)$ for every $x$ and $y$.

We turn next to R$\acute{e}$nyi  entropy. As it will turn out understanding the scaling 
properties of the R$\acute{e}$nyi  entropy will allow us immediately to understand the scaling
properties of the chi-squared and Hellinger divergences as well.
We have 
\[
\frac{1}{N} D_\alpha(P_N \mid \mid Q_N) = \frac{1}{N} \frac{1}{\alpha -1} \log \sum_{x_0, \cdots x_N}
\nu_q(x_0)^{1-\alpha}\nu_p(x_0)^\alpha \prod_{j=1}^N p(x_{j-1}, x_j)^{\alpha} q(x_{j-1}, x_j)^{1-\alpha}.   
\]
Let $F_\alpha$ be the non-negative matrix with entries 
\[
F_\alpha(x,y)= p(x,y)^{\alpha} q(x,y)^{1-\alpha}.
\]
since $p$ and $q$ are irreducible and mutually absolutely continuous the matrix 
$F_\alpha$ is irreducible as well. Let $v$ be the row vector with entries 
$v(x)=\nu_q(x)^{1-\alpha}\nu_p(x)$ and $1$ the column vector with all entries equal 
to $1$. Then we have 
\[
\frac{1}{N} D_\alpha(P_N \mid \mid Q_N) =  \frac{1}{\alpha-1} v F_\alpha^N 1 ,
\]
and thus by the Perron-Frobenius Theorem \cite{dembo2009large}, we have
\[
\lim_{N \to \infty} \frac{1}{N}D_\alpha(P_N,Q_N) = \frac{1}{\alpha-1} \log \rho(\alpha).
\]
where $\rho(\alpha)$ is the maximal eigenvalue of the non-negative matrix $F_\alpha$. 

It remains to show that the limit is $0$ only if $p=q$. In order to do this  
we will use some convexity properties of the R$\acute{e}$nyi  entropy 
\cite{van2014renyi}.  For $0< \alpha \le 1$  the R$\acute{e}$nyi  entropy $D_\alpha(P\mid 
\mid Q)$ is jointly convex in $P$ and $Q$, i.e. for any $\epsilon \in [0,1]$ we 
have 
\[
D_\alpha( \epsilon P_0 + (1- \epsilon) P_1 \mid \mid \epsilon 
Q_0 + (1- \epsilon) Q_1) \le \epsilon D(P_0\mid \mid Q_0) + (1 - \epsilon) D(P_1\mid 
\mid Q_1).
\]
For $\alpha> 1$ the R$\acute{e}$nyi  entropy is merely jointly quasi-convex, that is 
\[
D_\alpha(\epsilon P_0 + (1- \epsilon) P_1 \mid \mid \epsilon Q_0 + (1- 
\epsilon) Q_1) \le \max \left\{ D(P_0\mid \mid Q_0), D(P_1\mid \mid Q_1) \right\}.
\] 
In any case let us assume that $p \not= q$ is such that 
\[
\lim_{N \to \infty} \frac{1}{N} D_\alpha(P_N\mid \mid Q_N) = 0 \,. 
\]
Then by convexity, or quasi-convexity we have for any $\epsilon \in [0,1]$  
\[
\lim_{N \to \infty}\frac{1}{N} D_\alpha( \epsilon P_N + (1-\epsilon) Q_N \mid \mid Q_N) = 0 \,. 
\]
On the other hand, for any smooth parametric family $Q_\theta$ we have that, \cite{van2014renyi}, 
\[
D_\alpha(Q^{\theta'} \mid \mid Q^\theta) = \frac{\alpha}{2} (\theta - \theta')^2 J(Q^\theta) + O((\theta'-\theta)^3)  
\]
where $J(Q^\theta)$ is the Fisher information. If $Q^\theta$ is a discrete 
probability distribution then the Fisher information is $J(Q^\theta)= \sum_{x} 
Q^\theta(x) (\frac{d}{d\theta} \log Q^\theta(x))^2$.  

To compute $J(Q_N^\theta)$ we can use the relative entropy $R(Q_N^{\theta'} \mid \mid Q_N^\theta)=D_1(Q_N^{\theta'} \mid \mid Q_N^\theta)$  and from \eqref{eq:rerate} with $q=q^\theta$ and $p=q^{\theta'}$ 
we obtain 
\begin{eqnarray}
R(Q^{\theta'}_N \mid \mid Q^\theta_N) &=& (\theta'-\theta)^2 \sum_{x} \nu_{q^\theta}(x) \left(\frac{d}{d \theta} \log \nu_{q^\theta}(x)\right)^2 \nonumber \nonumber \\ 
&& +\frac{1}{2} (\theta-\theta')^2 \sum_{k=1}^N \sum_{x,y} (\nu_{q^\theta})^k(x) q^\theta(x,y)  \left(\frac{d}{d\theta} \log q^\theta(x,y)\right)^2 + O((\theta-\theta')^3). \nonumber
\end{eqnarray}
So as $N \to \infty$ we obtain 
\begin{equation}
\lim_{N \to \infty} \frac{1}{N} R(Q_N^{\theta'} \mid \mid Q_N^\theta) = \frac{1}{2}(\theta-\theta')^2 \sum_{x,y} \mu_{q^\theta}(x) q^\theta(x,y)\left(\frac{d}{d\theta} \log q^\theta(x,y)\right)^2+O((\theta-\theta')^3)
\end{equation}
If we now apply this to the family $Q^\epsilon = Q_N + \epsilon(P_N - Q_N)$ we have
that 
\[
\lim_{N \to \infty} \frac{1}{N} R(Q_N + \epsilon(P_N-Q_N) \mid \mid Q_N) = \frac{1}{2}\epsilon^2 
\sum_{x,y} \mu_q(x) \frac{ (p(x,y) - q(x,y))^2}{q(x,y)} + O(\epsilon^3)
\]
since the term of order $\epsilon^2$ is strictly positive unless $p=q$ this 
contradicts our assumption that  $\lim_{N \to \infty}\frac{1}{N} D_\alpha( \epsilon P_N + (1-\epsilon) Q_N \mid \mid Q_N) = 0$. 

We can now easily deduce the scaling of the $\chi^2$ divergence from the R$\acute{e}$nyi  
relative entropy because  of the relation $\chi^2(Q_N \mid \mid P_N) = e^{D_2(Q_N \mid \mid P_N)}-1$. This implies that $\chi^2(Q_N \mid \mid P_N)$ grows exponentially in $N$ unless 
$\lim_{N\to \infty} \frac{1}{N}D_2(Q_N \mid \mid P_N)=0$ which is possible if and 
only if $p=q$.  

Similarly for the Hellinger distance we use the relation $H^2(P_N, Q_N) = 2 -2 e^{-\frac{1}{2}  D_{\frac{1}{2}}(Q_N \mid \mid P_N)}$ and the scaling of the R$\acute{e}$nyi  entropy to see $H(Q_N,P_N)$  converges to $\sqrt{2}$ unless $p=q$. This concludes the proof of Lemma \ref{lem:scaling_markov}.

\section{Background for Section 5}
\subsection{Ising models and mean field models}
\label{section:backgroup_Ising and mean}
\medskip 
\noindent 
{\bf One-dimensional Ising model} Consider an Ising model on the lattice $\Lambda_N$ which  a line of N sites, labelled successively $x=1,2,...,N$. At each site there is a spin $\sigma(x)$, with two possible values: $+1$ or $-1$. The the Hamiltonian is given by 
\begin{equation}
H_{N}(\sigma_{\Lambda_N})=-\beta J\sum_{x=1}^{N-1}\sigma(x)\sigma(x+1)-\beta h\sum_{x\in \Lambda_N}\sigma(x).
\end{equation}
The configuration probability is given by the Boltzmann distribution with inverse temperature $\beta  \ge 0$:
\begin{equation}
d\mu_N(\sigma_{\Lambda_N})=\frac{1}{Z_{N}}e^{- H_{N}(\sigma_{\Lambda_N})}dP_N(\sigma_{\Lambda_N}),
\end{equation}
where 
\begin{equation}
Z_{N}=\sum_{\sigma_{\Lambda_N}} e^{- H_{N}(\sigma_{\Lambda_N})}
\label{eq:partition_Ising1}
\end{equation} 
is the partition function  and
$P_N(\sigma_{\Lambda_N})$ is the counting measure on $\Lambda_N$.

By \cite{baxter2007exactly}, the magnetization is
\begin{equation}
M= \frac{e^{J\beta}\sinh(h\beta)}{\sqrt{e^{2J\beta}\sinh^2(h\beta)+e^{-2J\beta}}},
\label{eq:mean_Ising1}
\end{equation}

and the pressure is 
\begin{equation}
 P=\lim \limits_{N}\frac{1}{N}\log Z_{N} 
= \log[e^{\beta J}\cosh(\beta h)+\sqrt{e^{2J\beta}\sinh^2(h\beta)+e^{-2J\beta}}].
\label{eq:pressure_Ising1}
\end{equation}
Differentiating \eqref{eq:partition_Ising1} with respect to $J$ and using \eqref{eq:pressure_Ising1}, one obtain
\begin{equation}
\lim \limits_{N}\frac{1}{N}E_{\mu_N}[\sum_{x \in \Lambda_N}\sigma(x)\sigma(x+1)]=\lim \limits_{N}\frac{1}{\beta}\frac{\partial}{\partial J}(\frac{1}{N}\log Z_{N}) 
= 1-\frac{1}{k_1}\frac{2e^{-2\beta J}}{e^{\beta J}\cosh(\beta h)+k_1},
\label{eq:corre_Ising1}
\end{equation}
where
\begin{equation}
k_1=\sqrt{e^{2J\beta}\sinh^2(h\beta)+e^{-2J\beta}}.
\end{equation}

Consider the susceptibility $\mathcal{X}$, by Section 1.7 in \cite{baxter2007exactly}, we have
\begin{equation}
\mathcal{X}=\frac{\partial M}{\partial h}=\beta\lim \limits_{N} \frac{1}{N}Var_{\mu_N}(\sum_{x \in \Lambda_N}\sigma(x)).
\end{equation}
Thus, by differentiating \eqref{eq:mean_Ising1} with respect to h, we obtain 
\begin{equation}\label{eq:var_Ising1}
\lim \limits_{N} \frac{1}{N}Var_{\mu_N}(\sum_{x \in \Lambda_N}\sigma(x))=\frac{e^{-J\beta}\cosh(h\beta)}{k_1^3}.
\end{equation}
 
\smallskip
\noindent
{\bf Square lattice zero-field Ising model} Consider an Ising model on the square lattice $\Lambda_N$ with $|\Lambda|=N$. Similarly with the 1-d Ising model, the spins $\{\sigma(x)\}_{x=1}^N \in \{-1,1\}^N$. Assume there is no external magnetic field, then Hamiltonian for the 2-d zero-field Ising model is given by 
\begin{equation}
H_{N}(\sigma_{\Lambda_N})=-\beta J\sum_{\langle x,y \rangle \subset \Lambda_N}\sigma(x)\sigma(y).
\label{eq:Hamiltonian_Ising2}
\end{equation}
where the first sum is over pairs of adjacent spins (every pair is counted once). The notation $\langle x,y\rangle $ indicates that sites $x$ and $y$ are nearest neighbors. 
Then the configuration probability is given by:
\begin{equation}
d\mu_{N}(\sigma_{\Lambda_N})=\frac{1}{Z_{N}}e^{\beta J\sum_{\langle x,y \rangle \subset \Lambda_N}\sigma(x)\sigma(y)}dP_N(\sigma_{\Lambda_N}),
\end{equation}
where 
\begin{equation}
Z_{N}=\sum_{\sigma_{\Lambda_N}} e^{\beta J\sum_{\langle x,y \rangle \subset \Lambda_N}\sigma(x)\sigma(y)}
\label{eq:partition_Ising2}
\end{equation} 
is the partition function  and $P_N(\sigma_{\Lambda_N})=\prod_{x=1}^NP(\sigma_{\Lambda_N})$ is the prior distribution with $P(\sigma(x)=1)=P(\sigma(x)=-1)=0.5$.
By Section 7.10 in \cite{baxter2007exactly}, the spontaneous  magnetization is 
\begin{equation} \label{eq:mean_Ising2}
M_0=
 \begin{cases}
    [1-\sinh^{-4}(2\beta J)]^{1/8}  &  \text{ $\beta>\beta_c$ ,} \nonumber \\
    0 
        &  \text{$\beta<\beta_c$,}
 \end{cases}                \end{equation}
 where $\beta_c=\frac{\log(1+\sqrt{2})}{2J}$.
 Actually, this formula for the spontaneous magnetization is given by the definition $
 M_0=\lim \limits_{h\to 0^{+}}\langle \sigma(x)\rangle$. Sometimes, we can also consider the spontaneous magnetization by using the other definition $M=\lim\limits_{h\to )^{-}}\langle \sigma(x)\rangle$, which actually is the opposite of \eqref{eq:mean_Ising2}. 
 
 And the pressure is also given by \cite{baxter2007exactly}
\begin{equation}
P=\lim \limits_{N \to \infty} \frac{1}{N}\log Z_{N} 
=\frac{\log2}{2}+\frac{1}{2\pi}\int_0^{\pi}\log[\cosh^2(2\beta J)+k(\theta)]d\theta,
\label{eq:pressure_Ising2}
\end{equation}
where \begin{equation}
k(\theta)=\sqrt{\sinh^4(2\beta J)+1-2\sinh^2(2\beta J)\cos(2\theta)}.
\end{equation}
And, by \eqref{eq:partition_Ising2} and \eqref{eq:pressure_Ising2}, we obtain
\begin{equation}\label{eq:corre_Ising2}
\lim \limits_{N \to \infty}\frac{1}{N}E_{\mu_{N}}(\sum_{\langle x,y\rangle \subset \Lambda_N}\sigma(x)\sigma(y))=\frac{1}{\beta}\frac{\partial}{\partial J}(\frac{1}{N}\log Z_{N} )=  
\frac{\sinh(4\beta J)}{\pi}\int_0^{\pi} \frac{1}{k(\theta)}[1-\frac{1+\cos(2\theta)}{\cosh^2(2\beta J)+k(\theta)}]d\theta .
\end{equation}

\smallskip
\noindent
{\bf Mean field model}
Given the Lattice $\Lambda_N$ in d-dimension and set $|\Lambda|=N$, consider the Hamiltonian for d-dimensional Ising model
\begin{equation}
 H_{N}(\sigma_{\Lambda_N}) = -\beta J\sum_{\langle x,y \rangle \subset \Lambda_N}  \sigma(x)\sigma(y)-\beta h\sum_{x \in \Lambda_N}^N\sigma(x)\nonumber \\
= -\sum_{x \in \Lambda_N} \sigma(x)\{\frac{1}{2}\beta J\sum_y^{n.n} \sigma(y)+\beta h \},
\label{eq:H_Ising}
\end{equation}
where the first sum is over pairs of adjacent spins (every pair is counted once). The notation $\langle x,y\rangle $ indicates that sites $x$ and $y$ are nearest neighbors. 
And, $\{\sigma(x)\}_{x=1}^{N} \in \{-1,1\}^N$ are Ising spins.
Replace $\sum_y^{n.n} \sigma(y)$ by $\sum_y^{n.n} \langle \sigma(y) \rangle$ in \eqref{eq:H_Ising}, we obtain the mean field Hamiltonian
\begin{eqnarray}
 H_{N;mf}(\sigma_{\Lambda_N}) & = & -\sum_{x \in \Lambda_N} \sigma(x)\{\frac{1}{2}\beta J\sum_y^{n.n} \langle \sigma(y) \rangle +\beta h \} \nonumber \\ 
& = & -\sum_{x \in \Lambda_N} \sigma(x)\{\frac{1}{2}\beta J 2d m +\beta h \} \nonumber \\ 
& = & -\sum_{x \in \Lambda_N} \sigma(x)\{\beta Jd m +\beta h \} \nonumber \\ 
& = & -\beta h_{mf}\sum_{x \in \Lambda_N}\sigma(x) \nonumber \\
\end{eqnarray}
where $h_{mf}=h+Jdm$.
Then, we have the probability
\begin{eqnarray}
d\mu_{N;mf}(\sigma_{\Lambda_N})=\frac{1}{Z_{N;mf}}e^{- H_{N;mf}(\sigma_{\Lambda_N})}dP_N(\sigma_{\Lambda_N})=\frac{1}{Z_{N;mf}}e^{\beta \sum_{x \in \Lambda_N}h_{mf}\sigma(x)}dP_N(\sigma_{\Lambda_N}).
\end{eqnarray}
So the partition function is 
\begin{eqnarray}
Z_{N;mf} & = &\sum_{\sigma(x)} e^{ \beta\sum_{x \in \Lambda_N}h_{mf}\sigma(x)} \nonumber \\
& = & \sum_{\sigma(x)} \prod_{x \in \Lambda_N} e^{\beta h_{mf}\sigma(x)} \nonumber \\
& = & \prod_{x \in \Lambda_N}\sum_{\sigma(x)}  e^{\beta h_{mf}\sigma(x)} \nonumber \\
& = & \prod_{x \in \Lambda_N} (e^{\beta h_{mf}}+e^{-\beta h_{mf}}) \nonumber \\
& = & (e^{\beta h_{mf}}+e^{-\beta h_{mf}})^N \nonumber \\
& = & {Z_{1;mf}}^N,
\end{eqnarray}
where $Z_{1;mf}=e^{\beta h_{mf}}+e^{-\beta h_{mf}}$.
So the pressure is
\begin{equation}\label{eq:pressure_mf}
P_{mf}=\lim\limits_{N}\frac{1}{N}\log Z_{N;mf}=\log(e^{\beta h_{mf}}+e^{-\beta h_{mf}})
\end{equation}
And, we can also consider the $\mu_{N;mf}$ as a product measure
\begin{eqnarray}
d\mu_{N;mf}(\sigma_{\Lambda_N})=\frac{1}{Z_{N;mf}}e^{\beta  \sum_{x \in \Lambda_N}h_{mf}\sigma(x)}dP_N(\sigma_{\Lambda_N})=\prod_{x \in \Lambda_N}\frac{1}{Z_{1;mf}}e^{\beta h_{mf}\sigma(x)}dP(\sigma(x)).
\end{eqnarray}
It is easy to find the magnetization
\begin{eqnarray}\label{eq:mean_mf}
m & = &\frac{1}{N}E_{\mu_{N;mf}}[\sum_{x \in \Lambda_N} \sigma(x)] \nonumber \nonumber \\
& = & \frac{1}{N}\sum_{x \in \Lambda_N}E_{\mu_{N;mf}}[\sigma(x)] \nonumber \nonumber \\ 
& = & \frac{1}{N}\sum_{x \in \Lambda_N}\sum_{\sigma(x)} \sigma(x)\frac{1}{Z_{1;mf}}e^{\beta h_{mf}\sigma(x)} \nonumber \nonumber \\
& = & \frac{1}{N}\sum_{x \in \Lambda_N}\sum_{\sigma(x)} \sigma(x)\frac{1}{Z_{1;mf}}e^{\beta h_{mf}\sigma(x)} \nonumber \nonumber \\
& = & \frac{1}{N}\sum_{x \in \Lambda_N} \frac{1}{Z_{1;mf}}(e^{\beta h_{mf}}-e^{-\beta h_{mf}}) \nonumber \nonumber \\
& = &  \frac{1}{e^{\beta h_{mf}}+e^{-\beta h_{mf}}}(e^{\beta h_{mf}}-e^{-\beta h_{mf}}) \nonumber \nonumber \\
& = & \tanh(\beta h_{mf})\nonumber \nonumber \\
& = & \tanh(\beta h+\beta Jdm)
\end{eqnarray}
and 
\begin{eqnarray}\label{eq:var_mf_appendix}
\frac{1}{N}Var_{\mu_{N;mf}}(\sum_{x \in \Lambda_N} \sigma(x)) & = &
\frac{1}{N}Var_{\mu_{N;mf}}(\sum_{x \in \Lambda_N} \sigma(x))\nonumber \\
& = & \frac{1}{N}(E_{\mu_{N;mf}}[\sum_{x \in \Lambda_N} \sigma(x)]^2 -N^2m^2) \nonumber \\
& = & \frac{1}{N}(E_{\mu_{N;mf}}[\sum_{x \in \Lambda_N} \sigma^2(x)+\sum_{x\ne y}\sigma(x)\sigma(y)] -Nm^2) \nonumber \\
& = & \{ \sum_{\sigma(x)} \sigma^2(x)\frac{1}{Z_{1;mf}}e^{\beta h_{mf}\sigma(x)}+(N-1)E_{\mu_{N;mf}}[\sigma(x)\sigma(y)]\} -Nm^2 \nonumber \\
& = & \{1 +(N-1)m^2\} -Nm^2 \nonumber \\
& = & 1- m^2.
\end{eqnarray}
So we can obtain the magnetization $m$ by solving the implicit equation \eqref{eq:mean_mf}. 
\subsection{Computation  of goal-oriented divergences}
\label{section: example_goal oriented divergence calculation }
\smallskip
\noindent
{\bf Mean field versus mean field} 
Given two mean field models, assume $\mu_{N;mf}$ and $\mu'_{N;mf}$ are their two configuration probabilities with 
\begin{eqnarray}
d\mu_{N;mf}(\sigma)=\frac{1}{Z_{N;mf}}e^{- H_{N;mf}(\sigma_{\Lambda_N})}dP_N(\sigma_{\Lambda_N})=\frac{1}{Z_{N;mf}}e^{\beta \sum_{x \in \Lambda_N}h_{mf}\sigma(x)}dP_N(\sigma_{\Lambda_N})
\end{eqnarray}
and 
\begin{eqnarray}
d\mu'_{N;mf}(\sigma)=\frac{1}{Z'_{N;mf}}e^{- H'_{N;mf}(\sigma_{\Lambda_N})}dP_N(\sigma_{\Lambda_N})=\frac{1}{Z'_{N;mf}}e^{\beta' \sum_{x \in \Lambda_N}h'_{mf}\sigma(x)}dP_N(\sigma_{\Lambda_N}).
\end{eqnarray}
where  $h_{mf}=h+dJm$ and $h'_{mf}=h'+dJ'm'$.
Then, by \eqref{eq:re_gibbs},
the relative entropy between $\mu'_{N;mf}$ and $\mu_{N;mf}$ is given by 
\begin{eqnarray}
R(\mu'_{N;mf}\|\mu_{N;mf})
&=& \log Z_{N;mf}-\log Z'_{N;mf}+E_{\mu'_{N;mf}}[H_{N;mf}(\sigma_{\Lambda_N})-H'_{N;mf}(\sigma_{\Lambda_N})]\nonumber \\
&=& \log \frac{Z_{N;mf}}{Z'_{N;mf}}+ (\beta' h'_{mf}-\beta h_{mf})E_{\mu'_{N;mf}}(\sum_{x \in \Lambda_N}\sigma(x))\nonumber \\
& = & N\log \frac{Z_{1;mf}}{Z'_{1;mf}}+(\beta' h'_{mf}-\beta h_{mf})Nm'\nonumber \\
& = & N\log \frac{e^{\beta h_{mf}} +e^{-\beta h_{mf}}}{e^{-\beta' h'_{mf}}+e^{\beta' h'_{mf}}}+N(\beta' h'_{mf}-\beta h_{mf})m'.
\end{eqnarray}
Therefore, we have
\begin{equation}
\frac{1}{N}R(\mu'_{N;mf}\|\mu_{N;mf})
=\log \frac{e^{\beta h_{mf}} +e^{-\beta h_{mf}}}{e^{-\beta' h'_{mf}}+e^{\beta' h'_{mf}}}+(\beta' h'_{mf}-\beta h_{mf})m'.
\end{equation}
And, the cumulant generating function of $Nf_N=N\frac{1}{N}\sum_{x \in \Lambda_N}\sigma(x)=\sum_{x \in \Lambda_N}\sigma(x)$ is 
\begin{eqnarray}
\Lambda_{\mu_{N;mf},Nf_N}(c) & = &\log E_{\mu_{N;mf}}(e^{cN\frac{1}{N}\sum_{x \in \Lambda_N} \sigma(x)}) \nonumber \\
 & = &\log \sum_{\sigma(x)}e^{cN\frac{1}{N}\sum_{x \in \Lambda_N} \sigma(x)}\frac{1}{Z_{N;mf}}e^{\beta h_{mf}\sum_{x \in \Lambda_N}\sigma(x)}\nonumber \\
 & = & \log \sum_{\sigma(x)}\frac{1}{Z_{N;mf}}e^{(c+\beta h_{mf})\sum_{x \in \Lambda_N}\sigma(x)} \nonumber \\
 & = & \log \sum_{\sigma(x)} \prod_{x \in \Lambda_N}\frac{1}{Z_{1;mf}}e^{(c+\beta h_{mf})\sigma(x)} \nonumber \\ 
 & = & \log \prod_{x \in \Lambda_N}\sum_{\sigma(x)} \frac{1}{Z_{1;mf}}e^{(c+\beta h_{mf})\sigma(x)} \nonumber \\
 & = & \log \prod_{x \in \Lambda_N}\frac{1}{Z_{1;mf}}\{e^{(c+\beta h_{mf})} +e^{-(c+\beta h_{mf})}\}\nonumber \\
 & = & N\log \frac{e^{(c+\beta h_{mf})} +e^{-(c+\beta h_{mf})}}{e^{-\beta h_{mf}}+e^{\beta h_{mf}}}.
\end{eqnarray}
Thus,
\begin{equation}
\frac{1}{N}\Lambda_{\mu_{N;mf},Nf_N}(c) =\log \frac{e^{(c+\beta h_{mf})} +e^{-(c+\beta h_{mf})}}{e^{-\beta h_{mf}}+e^{\beta h_{mf}}}.
\end{equation}
Also, by \eqref{eq:var_mf_appendix}, we have
\begin{eqnarray}
\frac{1}{N}Var_{\mu_{N;mf}}(Nf_N) 
& = & 1- m^2.
\end{eqnarray}

\smallskip
\noindent
{\bf One-dimensional Ising model versus  mean field } 
Consider the Ising model and mean field model in 1-d and assume $\mu_N$ and $\mu_{N;mf}$ are the configuration probabilities for 1-d Ising model and mean field model respectively, which are defined in section \ref{section:backgroup_Ising and mean}. 
Then, by \eqref{eq:re_gibbs}, the relative entropy between $\mu_N$ and $\mu_{N;mf}$ is 
\begin{align}
R(\mu_N\|\mu_{N;mf})&=\log Z_{N;mf}-\log Z_{N}+E_{\mu_N}(H_{N;mf}(\sigma_{\Lambda_N})-H_{N}(\sigma_{\Lambda_N}))\nonumber \\
&=\log Z_{N;mf}-\log Z_{N}+\beta JE_{\mu_N}(\sum_{\langle x,y\rangle \subset \Lambda_N}\sigma(x)\sigma(y))-\beta Jm E_{\mu_N}(\sum_{x\in\Lambda_N}\sigma(x)).
\end{align}
Thus, by \eqref{eq:pressure_mf}, \eqref{eq:pressure_Ising1}, \eqref{eq:corre_Ising1} and \eqref{eq:mean_mf}, we have
\begin{align}
  &\lim\limits_{N}\frac{1}{N}R(\mu_N\|\mu_{N;mf}) \nonumber \\
 = & \lim\limits_{N}\frac{1}{N}\log Z_{N;mf}-\lim\limits_{N}\frac{1}{N}\log Z_{N}+\beta J\lim\limits_{N}\frac{1}{N}E_{\mu_N}(\sum_{\langle x,y\rangle \subset \Lambda_N}\sigma(x)\sigma(y))
-\lim\limits_{N}\beta Jm\frac{1}{N}E_{\mu_N}(\sum_{x\in\Lambda_N}\sigma(x))\nonumber \\
=& \log \frac{e^{\beta [h+Jm]}+e^{-\beta [h+Jm]}}{e^{\beta J}\cosh(\beta h)+k_1}+\frac{\beta J}{k_1}(k_1-\frac{2e^{-2\beta J}}{e^{\beta J}\cosh(\beta h)+k_1}-me^{J\beta}\sinh(h\beta))
\end{align}

And, by \eqref{eq:cumulant_mf} and by \eqref{eq:var_mf},we obtain
\begin{equation}
\frac{1}{N}\Lambda_{\mu_{N;mf},Nf_N}(c) 
  =  \log \frac{e^{[c+\beta (h+Jm)] } +e^{-[c+\beta (h+Jm)]}}{e^{-\beta [h+Jm]}+e^{\beta [h+Jm]}}
\end{equation}
and
\begin{equation}
\frac{1}{N}Var_{\mu_{N;mf}}(Nf_N) .
= 1- m^2.
\end{equation}

\smallskip
\noindent
{\bf Two-dimensional  Ising model with $h=0$ versus  mean field } 
Assuming $\mu_N$ and $\mu_{N;mf}$ are two configuration probabilities for two-dimensions zeros Ising model and two-dimensions zeros mean field model respectively.
By Section \ref{section:backgroup_Ising and mean},
\begin{equation}
\mu_N(\sigma_{\Lambda_N})=\mu_{N}(\sigma_{\Lambda_N})=\frac{1}{Z_{N}}e^{- H_{N}(\sigma_{\Lambda_N})}P_N(\sigma_{\Lambda_N})=\frac{1}{Z_{N}}e^{\beta J\sum_{\langle x,y \rangle \subset \Lambda_N}\sigma(x)\sigma(y)}P_N(\sigma_{\Lambda_N})
\end{equation}
and
\begin{equation}
\mu_{N;mf}(\sigma_{\Lambda_N})=\frac{1}{Z_{N;mf}}e^{- H_{N;mf}(\sigma_{\Lambda_N})}d\sigma=\frac{1}{Z_{N;mf}}e^{\beta \sum_{x \in \Lambda_N}h_{mf}\sigma(x)}P_N(\sigma_{\Lambda_N}),
\end{equation}
where $Z_{N;mf}=(e^{\beta h_{mf}}+e^{-\beta h_{mf}})^N$ and $h_{mf}=2Jm$.

Then, by \eqref{eq:re_gibbs}, the relative entropy between $\mu_N$ and $\mu_{N;mf}$ is
\begin{align}
R(\mu_N\|\mu_{N;mf})&=\log Z_{N;mf}-\log Z_{N}+E_{\mu_N}(H_{N;mf}(\sigma_{\Lambda_N})-H_{N}(\sigma_{\Lambda_N}))\nonumber \\
&=\log Z_{N;mf}-\log Z_{N}+\beta JE_{\mu_N}(\sum_{\langle x,y\rangle \subset \Lambda_N}\sigma(x)\sigma(y))-2\beta Jm E_{\mu_N}(\sum_{x\in\Lambda_N}\sigma(x)).
\end{align}
Thus, by \eqref{eq:pressure_mf}, \eqref{eq:pressure_Ising2}, \eqref{eq:corre_Ising2} and \eqref{eq:mean_mf}, we have
\begin{align}
  &\lim\limits_{N}\frac{1}{N}R(\mu_N\|\mu_{N;mf})\nonumber\\ 
 = & \lim\limits_{N}\frac{1}{N}\log Z_{N;mf}-\lim\limits_{N}\frac{1}{N}\log Z_{N}+\beta J\lim\limits_{N}\frac{1}{N}E_{\mu_N}(\sum_{\langle x,y\rangle \subset \Lambda_N}\sigma(x)\sigma(y))
-\lim\limits_{N}2\beta Jm\frac{1}{N}E_{\mu_N}(\sum_{x\in\Lambda_N}\sigma(x))\nonumber \\
=& \log \frac{e^{\beta [h+Jm]}+e^{-\beta [h+Jm]}}{e^{\beta J}\cosh(\beta h)+k_1}+\frac{\beta J}{k_1}(k_1-\frac{2e^{-2\beta J}}{e^{\beta J}\cosh(\beta h)+k_1}-me^{J\beta}\sinh(h\beta))\nonumber \\
=&\log [e^{-2\beta Jm}+e^{2\beta Jm}]-\frac{\log2}{2}-\frac{1}{2\pi}\int_0^{\pi}\log[\cosh^2(2\beta J)+k(\theta)]d\theta \nonumber\nonumber \\
 &+\beta J \frac{\sinh(4\beta J)}{\pi}\int_0^{\pi} \frac{1}{k(\theta)}[1-\frac{1+\cos(2\theta)}{\cosh^2(2\beta J)+k(\theta)}]d\theta-2\beta JmM
\end{align}


And, by \eqref{eq:cumulant_mf} and by \eqref{eq:var_mf_appendix}, we obtain
\begin{equation}
\frac{1}{N}\Lambda_{\mu_{N;mf},Nf_N}(c) 
  =  \log \frac{e^{(c+2\beta Jm) } +e^{-(c+2\beta Jm)}}{e^{-2\beta Jm}+e^{2\beta Jm}}
\end{equation}
and
\begin{equation}
\frac{1}{N}Var_{\mu_{N;mf}}(Nf_N) .
= 1- m^2.
\end{equation}

\smallskip
\noindent
{\bf One-dimensional Ising model versus Ising model} 
Consider two Ising models in 1-d and $\mu_N$ and $\mu_N'$ are their configuration probabilities defined in Section\ref{section:backgroup_Ising and mean}.
Then, by \eqref{eq:pressure_Ising1}, \eqref{eq:corre_Ising1} and \eqref{eq:mean_Ising1}, we have
\begin{align}
&\lim\limits_{N}\frac{1}{N}R(\mu'_N\|\mu_N)\nonumber \\ 
= &\lim\limits_{N}\frac{1}{N}E_{\mu'_N} (\log\frac{\mu'_N}{\mu_N})\nonumber \\
 = & \lim\limits_{N}\frac{1}{N}\log\frac{Z_{N}}{Z'_{\Lambda,N}}+\lim\limits_{N}\frac{1}{N}E_{\mu'_N}(H(\sigma_{\Lambda_N})-H'_{N}(\sigma_{\Lambda_N}))\nonumber \\
= & \lim\limits_{N}\frac{1}{N}\log Z_{N}-\lim\limits_{N}\frac{1}{N}\log Z'_{\Lambda,N}+(\beta' J'-\beta J)\lim\limits_{N}\frac{1}{N}E_{\mu'_N}(\sum_{\langle x,y\rangle \subset \Lambda_N}\sigma(x)\sigma(y))\nonumber\nonumber \\
&+ (\beta' h'-\beta h)\lim\limits_{N}\frac{1}{N}E_{\mu'_N}(\sum_{x \in \Lambda_N}\sigma(x))\nonumber \\
=&\log\frac{e^{\beta J}\cosh(\beta h)+\sqrt{e^{2J\beta}\sinh^2(h\beta)+e^{-2J\beta}}}{e^{\beta' J'}\cosh(\beta' h')+\sqrt{e^{2J'\beta'}\sinh^2(h'\beta')+e^{-2J'\beta'}}}+ (\beta' J'-\beta J)(1-\frac{1}{k'_1}\frac{2e^{-2\beta' J'}}{e^{\beta' J'}\cosh(\beta' h')+k'_1})\nonumber\nonumber \\
&+(\beta' h'-\beta h)
\frac{1}{k'_1}e^{J'\beta'}\sinh(h'\beta')
\end{align}
And,
\begin{eqnarray}
\lim \limits_{N}\frac{1}{N}\Lambda_{\mu_N}(c) & = &\lim \limits_{N}\frac{1}{N}\log E_{\mu_{N}}(e^{cN\frac{1}{N}\sum_{x \in \Lambda_N} \sigma(x)}) \nonumber \\
 & = &\lim \limits_{N}\frac{1}{N}\log \sum_{\sigma_{\Lambda_N}}e^{cN\frac{1}{N}\sum_{x \in \Lambda_N} \sigma(x)}\frac{1}{Z_{N}}e^{\beta J \sum_{\langle x,y \rangle \subset \Lambda_N}\sigma(x)\sigma(y)+ \beta h\sum_{x \in \Lambda_N}\sigma(x)}\nonumber \\
 & = &\lim \limits_{N}\frac{1}{N}\log \frac{1}{Z_{N}} \sum_{\sigma_{\Lambda_N}}e^{\beta J \sum_{\langle x,y \rangle \subset \Lambda_N}\sigma(x)\sigma(y)+ \beta (h +\frac{c}{\beta})\sum_{x \in \Lambda_N}\sigma(x)}\nonumber \\
 & = &\lim \limits_{N}\frac{1}{N}\log \frac{1}{Z_{N}}\tilde{Z}_{\Lambda,N} \nonumber \\
  & = &\lim \limits_{N}\frac{1}{N}\log \tilde{Z}_{\Lambda,N} -\lim \limits_{N}\frac{1}{N}\log Z_{N},
\end{eqnarray}
where $\tilde{Z}_{\Lambda,N}=\sum_{\sigma_{\Lambda_N}}e^{\beta J \sum_{\langle x,y \rangle \subset \Lambda_N}\sigma(x)\sigma(y)+ \beta (h +\frac{c}{\beta})\sum_{x \in \Lambda_N}\sigma(x)}.$
By \cite{baxter2007exactly}, we have
\begin{equation}
\lim \limits_{N}\frac{1}{N}\log \tilde{Z}_{\Lambda,N}=\log [e^{\beta J}\cosh(\beta h+c)+\sqrt{e^{2J\beta}\sinh^2(h\beta+c)+e^{-2J\beta}}]
\end{equation}
and
\begin{equation}
\lim \limits_{N}\frac{1}{N}\log Z_{N}=\log [e^{\beta J}\cosh(\beta h)+\sqrt{e^{2J\beta}\sinh^2(h\beta)+e^{-2J\beta}}].
\end{equation}
Thus,
\begin{equation}
\lim \limits_{N}\frac{1}{N}\Lambda_{\mu_N}(c) 
 =\log \frac{e^{\beta J}\cosh(\beta h+c)+\sqrt{e^{2J\beta}\sinh^2(h\beta+c)+e^{-2J\beta}}}{e^{\beta J}\cosh(\beta h)+\sqrt{e^{2J\beta}\sinh^2(h\beta)+e^{-2J\beta}}}.
\end{equation}
And, by \eqref{eq:var_Ising1}
\begin{equation}
\lim \limits_{N} \frac{1}{N}Var_{\mu_N}(\sum_{x \in \Lambda_N}\sigma(x))=\frac{e^{-J\beta}\cosh(h\beta)}{k_1^3}.
\end{equation}
\end{document}